\documentclass [10pt,a4paper]{amsart}
\usepackage{mathtools}
\usepackage{mathrsfs}
\usepackage{amssymb}
\usepackage{latexsym}
\usepackage{yfonts}
\usepackage{natbib}

\usepackage[pdftex,plainpages=false,colorlinks,hyperindex,bookmarksopen,linkcolor=red,citecolor=blue,urlcolor=blue]{hyperref}

\bibpunct{[}{]}{;}{n}{,}{,}

\newtheorem{te}{Theorem}[section]

\theoremstyle{definition}
\newtheorem{definition}[te]{Definition}

\theoremstyle{os}
\newtheorem{os}[te]{Remark}

\theoremstyle{prop}
\newtheorem{prop}[te]{Proposition}
\newtheorem{example}{Example}
\theoremstyle{lem}

\theoremstyle{coro}

\numberwithin{equation}{section}

\allowdisplaybreaks

\begin{document}

\title[Fractional directional gradient and its application]{Fractional gradient and its application to the fractional advection equation}

\author{Mirko D'Ovidio}
\email{mirko.dovidio@uniroma1.it}

\author{Roberto Garra}
\email{roberto.garra@sbai.uniroma1.it}

\address{Dipartimento di Scienze di Base e Applicate per l'Ingegneria, Sapienza University of Rome}


\date{\today}

\begin{abstract}
In this paper we provide a definition of fractional gradient operators, related to directional derivatives. We develop
a fractional vector calculus, providing a probabilistic interpretation and mathematical tools to treat multidimensional fractional differential equations.
A first application is discussed in relation to the d-dimensional fractional advection-dispersion equation. We also study the connection with multidimensional L\'evy processes.

\end{abstract}

\maketitle

\section{Introduction}

Fractional calculus is a developing field of the applied mathematics
regarding integro-differential equations involving fractional
integrals and derivatives. The increasing interest in fractional calculus
has been motivated by many applications of fractional equations
in different fields of research (see for example
\cite{Caputo,Main,Met,Ors}). However, most of the papers in this field are focused
on the analysis of fractional equations and processes in one
dimension, there are few works regarding fractional vector calculus
and its applications in theory of electromagnetic fields,
fluidodynamics and multidimensional processes. A first attempt to
give a formulation of fractional vector calculus is due to Ben
Adda(\cite{Ben Adda}). Recently a different approach in the
framework of multidimensional fractional advection-dispersion
equation has been developed by Meerschaert et al. (\cite{m1}, \cite{m2},
\cite{m3}). They present a general definition of gradient, divergence
and curl, in relation to fractional directional derivatives. In
their view, the fractional gradient is a weighted sum of fractional
directional derivatives in each direction. We notice that this
general approach to fractional gradient, depending on the choice of
the mixing measure, includes also the definition of fractional
gradient given by Tarasov in \cite{T2},\cite{T3}
(see also the recent book \cite{T1}) as a natural extension of the
ordinary case. Starting from these works, many authors have been
interested in understanding the applications of this fractional vector
calculus
in the theory of electromagnetic fields in fractal media (see for example \cite{ZAMP} and \cite{Baleanu}) and in the
analysis of multidimensional advection-dispersion equation (\cite{Bol, m3}).
Moreover, in \cite{silvia}, the authors study the application of fractional vector calculus to
the multidimensional Bloch-Torrey equation.\\

We study some equations involving a slightly modified version of the
fractional gradient introduced by Meerschaert et al. in \cite{m3} and we provide a new
class of fractional power of operators based on such gradient. We show some consequences of our approach, to treat multidimensional fractional differential equations.
In particular we discuss a first application of the fractional gradient to the fractional advection and dispersion equation and we find deterministic and stochastic solutions. We present a formulation of this equation in relation to the fractional conservation of mass, introduced by Meerschaert et al. in \cite{mass}.\\
Furthermore,
the properties of a class of multidimensional L\'evy processes related to fractional gradients are
investigated. Indeed, we introduce a novel L\'evy-Khinchine formula, involving our fractional gradient,
for the infinitesimal generator of a multidimensional random process.
It is well known that long jump random walks lead to limit processes
governed by the fractional Laplacian. We establish
some connections between compound Poisson processes with given jumps 
and the corresponding limit processes which are driven 
by our new L\'evy-Khinchine formula involving fractional gradient.\\
A general translation semigroup and the related 
Frobenius-Perron operator are also introduced and the associated 
advection equations are investigated. As in the previous cases, we find relation with 
compound Poisson processes.\\
We finally study the fractional power of the second order directional derivative $\left(\boldsymbol{\theta}\cdot\nabla\right)^2$ and the heat-type equation involving this operator.

\section{Fractional gradient operators and fractional directional derivatives}

In the general approach developed by Meerschaert et al.(\cite{m3}) in the framework of the multidimensional fractional advection-dispersion equation,
given a scalar function $f(\mathbf{x})$, the fractional gradient can be defined as
\begin{equation}\label{M}
\nabla^{\beta}_M f(\mathbf{x})=\int_{||\boldsymbol{\theta}||=1}\boldsymbol{\theta}D_{\boldsymbol{\theta}}^{\beta}f(\mathbf{x})M(d\boldsymbol{\theta}),\; \mathbf{x}\in\mathbb{R}^d, \beta\in(0,1)
\end{equation}
where $\boldsymbol{\theta}=(\theta_1,....,\theta_d)$ is a unit column vector; $M(d\boldsymbol{\theta})$ is a positive finite measure, called mixing measure;
\begin{equation}\label{MM}
D_{\boldsymbol{\theta}}^{\beta}f(\mathbf{x})=(\boldsymbol{\theta}\cdot \nabla)^{\beta}f(\mathbf{x}),
\end{equation}
is the fractional directional derivative of order $\beta$ (see for example \cite{Mirko}).\\
The Fourier transform of fractional directional derivatives
\eqref{MM} (in our notation) is given by
\begin{equation}
\widehat{D_{\boldsymbol{\theta}}^{\beta}f}(\mathbf{k})=\widehat{(\boldsymbol{\theta}
\cdot \nabla)^\beta f}(\mathbf{k}) = (- i \, \boldsymbol{\theta}
\cdot \mathbf{k})^\beta \, \widehat{f}(\mathbf{k}),
\end{equation}
where
$$\widehat{f}(\mathbf{k})=\int_{\mathbb{R}^d} e^{i\mathbf{k}\cdot \mathbf{x}}f(\mathbf{x})d\mathbf{x}.$$
Hence the Fourier transform of \eqref{M} is
written as
\begin{equation}\label{TF}
\widehat{\nabla^{\beta}_M f}(\mathbf{k})=
\int_{||\boldsymbol{\theta}||=1}\boldsymbol{\theta}(-i\mathbf{k}\cdot
\boldsymbol{\theta})^{\beta}\widehat{f}(\mathbf{k})M(d\boldsymbol{\theta}).
\end{equation}
This is a general definition of fractional gradient, depending on the choice of the mixing measure $M(d\boldsymbol{\theta})$. We can infer the physical and
geometrical meaning of this definition: it is a weighted sum of the fractional directional derivatives in each direction on a unitary sphere.\\
The definition \eqref{M} is really general and directly related to
multidimensional stable distributions. The divergence of \eqref{M} is given by
\begin{equation}\label{D}
\mathbb{D}_M^{\alpha}f(\mathbf{x}):= \nabla \cdot
\nabla^{\alpha-1}_M f(\mathbf{x})=
\int_{||\boldsymbol{\theta}||=1}D_{\boldsymbol{\theta}}^{\alpha}f(\mathbf{x})M(d\boldsymbol{\theta}),\;
\mathbf{x}\in\mathbb{R}^d, \alpha\in(1,2],
\end{equation}
whose Fourier transform, from \eqref{TF}, is written as
\begin{equation}
\widehat{\mathbb{D}^{\alpha}_Mf}(\mathbf{k})=
\int_{||\boldsymbol{\theta}||=1}(-i\mathbf{k}\cdot\boldsymbol{\theta})^{\alpha}f(\mathbf{k})M(d\boldsymbol{\theta}).
\end{equation}
The scalar operator $\mathbb{D}_M^{\alpha}$ plays the role of fractional Laplacian in the fractional diffusion equation,
introducing a more general class of processes depending on the choice of the measure $M$.\\
For the sake of clarity we
refer to Meerschaert et al. (\cite{m1}) about
multidimensional fractional diffusion-type equations
involving this kind of operators.
Let us consider the multidimensional fractional diffusion-type equation involving $\mathbb{D}_M^{\alpha}$, given
by
\begin{equation}\label{dif1}
\frac{\partial u}{\partial t}(\mathbf{x},t)= \mathbb{D}_M^{\alpha}u(\mathbf{x},t),
\end{equation}
with initial condition
$$u(\mathbf{x},0)=\delta(\mathbf{x}).$$
We obtain by Fourier transform
\begin{equation}
\frac{\partial \hat{u}}{\partial t}(\mathbf{k},t)=
\int_{||\boldsymbol{\theta}||=1}(-i\mathbf{k}\cdot\boldsymbol{\theta})^{\alpha}M(d\boldsymbol{\theta})\hat{u}
(\mathbf{k},t).
\end{equation}
Then, the the solution of \eqref{dif1} in the Fourier space is given by
\begin{equation}
\hat{u}(\mathbf{k},t)=
exp\left(t\int_{||\boldsymbol{\theta}||=1}(-i\mathbf{k}\cdot\boldsymbol{\theta})^{\alpha}M(d\boldsymbol{\theta})\right),
\end{equation}
which is strictly related with multivariate stable distributions, as the following well known result
entails

\begin{te}[\cite{Taqqu}, pag.65]

Let $\alpha \in (0,2)$, then $\boldsymbol{\theta}= (\theta_1,...,\theta_d)$ is an $\alpha$-stable random vector in $\mathbb{R}^d$ if and only if there exists a finite
measure $\Gamma$ on the unitary sphere and a vector $\boldsymbol{\mu}^0=(\mu_1^0, ....\mu_d^0)$ such that its characteristic function is given by\\
$$\mathbb{E}exp\{i(\mathbf{k}\cdot\boldsymbol{\theta})\}= e^{-\sigma \psi(\mathbf{k})},$$
where $\sigma = \cos(\frac{\pi \alpha}{2})$, and
\begin{equation}\nonumber
\psi(\mathbf{k})=
\begin{cases}
\int_{||\boldsymbol{\theta}||=1}|\boldsymbol{\theta}\cdot \mathbf{k}|^{\alpha}(1-i sign
(\boldsymbol{\theta}\cdot\mathbf{k})\tan\frac{\pi\alpha}{2})
\Gamma(d\boldsymbol{\theta})+i(\mathbf{k}\cdot\boldsymbol{\mu}^0), & \mbox{if } \alpha \neq 1, \\
\int_{||\boldsymbol{\theta}||=1}|\boldsymbol{\theta}\cdot \mathbf{k}|(1+i\frac{2}{\pi} sign
(\boldsymbol{\theta}\cdot\mathbf{k})\ln|(\boldsymbol{\theta}\cdot\mathbf{k})|)
\Gamma(d\boldsymbol{\theta})+i(\mathbf{k}\cdot\boldsymbol{\mu}^0), & \mbox{if } \alpha = 1.
\end{cases}
\end{equation}
The pair $(\Gamma, \boldsymbol{\mu}^0)$ is unique.
\end{te}

In light of Theorem 2.1 and the fact that
\begin{equation}
(-i\zeta)^{\alpha}=
|\zeta|^{\alpha}e^{-i\frac{\pi}{2}\alpha\frac{\zeta}{|\zeta|}}=
|\zeta|^{\alpha}e^{-i\frac{\pi}{2}\alpha sign(\zeta)},
\end{equation}
 the solution of
\eqref{dif1} can be interpreted as the law of a $d$-dimensional
$\alpha$-stable vector, whose characteristic function is given, for
$\alpha\neq 1$, by the pair $(M,\mathbf{0})$, i.e. the vector
$\boldsymbol{\mu}^0$ is null and the measure $M$ is the spectral
measure of the random vector $\boldsymbol{\theta}$. This is a general approach to multidimensional fractional
differential equations, suggesting the geometrical and
probabilistic meaning of \eqref{dif1}. On the other hand it includes
a wide class of processes, depending on the spectral measure $M$. As
a first notable example, being
$M(d\boldsymbol{\theta})=
m(\boldsymbol{\theta})d\boldsymbol{\theta}$, if we take
$m(\boldsymbol{\theta})= const.$ in \eqref{D}, then we obtain the
well known Riesz derivative (\cite{libro}). In the framework of
fractional vector calculus we obtain a
geometric interpretation of the fractional Laplacian which is strictly related to uniform isotropic measure.\\
We also notice that the definition of fractional gradient given by
Tarasov (\cite{T2}) is a special case of \eqref{M}, corresponding to
the case in which the mixing measure is a point mass at each
coordinate vector $\mathbf{e}_i$, for $i= 1, ...., d$. In this case
the fractional gradient seems to be a formal extension of the
ordinary to the fractional case, i.e.
\begin{equation}\label{gtar}
\nabla^{\beta}f(\mathbf{x})=\sum_{i=1}^{d}\frac{\partial^{\beta} f(\mathbf{x})}{\partial x_i^\beta}\mathbf{e}_i,
\end{equation}
where $\frac{\partial^{\beta} f}{\partial x_i^\beta}$ is the Weyl partial fractional derivative of order $\beta\in(0,1)$, defined as (\cite{libro})
\begin{equation}
\frac{d^{\beta} f}{d x^\beta}= \frac{1}{\Gamma(1-\beta)}\frac{d}{d x}\int_{-\infty}^x (x-y)^{-\beta}f(y)dy,\quad x\in\mathbb{R}.
\end{equation}
Formula \eqref{gtar} seems to be a natural way to generalize the definition of gradient of fractional order.
Indeed, for $\beta =1$ we recover the ordinary gradient.
From a geometrical point of view this is an integration centered on preferred directions given by the Cartesian set of axes.
From a probabilistic point of view this is the unique case in which an $\alpha$-stable random vector has independent components as shown by Samorodnitsky
and Taqqu (\cite{Taqqu}, Example 2.3.5, pag.68). It corresponds to a choice of the spectral measure $\Gamma$ discrete and concentrated on the intersection of
the axes with the unitary sphere.\\

In this paper we adopt an intermediate approach between the special
case treated by Tarasov (\cite{T1}) and the most general one treated
by Meerschaert et al. (\cite{m3}). Indeed, we consider the following
subcase of the general definition \eqref{M}

\begin{definition}

For $\beta \in(0,1)$ and a ``good'' scalar function $f(\mathbf{x})$,
$\mathbf{x}\in \mathbb{R}^d$, being
$(\boldsymbol{\theta}_1,.......,\boldsymbol{\theta}_d)$, with
$\boldsymbol{\theta}_j\in\mathbb{R}^d$, for $j = 1,2,..,d$, an
orthonormal basis, the fractional gradient is written as
\begin{equation}\label{ngr}
\nabla^{\beta}_{\theta}f(\mathbf{x})=\sum_{l= 1}^d
\boldsymbol{\theta}_l(\boldsymbol{\theta}_l\cdot
\nabla)^{\beta}f(\mathbf{x}), \quad f\in L^1(\mathbb{R}^d),
\end{equation}
where we use the subscript $\theta$ to underline the connection with the
mixing measure $M$ which is a point mass measure at each
coordinate vectors $\boldsymbol{\theta}_l$, $l= 1,\cdots, d$.

\end{definition}

\bigskip

This is a superposition of fractional directional derivatives, taking into account all the directions $\boldsymbol{\theta}_i$, it is a more general
approach than that adopted by Tarasov. However, also in this case, for $\beta = 1$ we recover the definition of the ordinary gradient.
An explicit representation of the fractional gradient \eqref{ngr} is given by means of operational methods. Indeed, in \cite{Mirko}, it was shown that the fractional power of the
directional derivative is given by
\begin{equation}
\left(\boldsymbol{\theta}\cdot\nabla\right)^{\beta}f(\mathbf{x})=\frac{\beta}{\Gamma(1-\beta)}
\int_0^{\infty}\left(f(\mathbf{x})-f(\mathbf{x}-s\boldsymbol{\theta})\right)s^{-\beta-1}ds, \quad \beta\in(0,1),
\end{equation}
so that \eqref{ngr} has the following representation
\begin{equation}
\nabla^{\beta}_{\theta}f(\mathbf{x})=\sum_{l= 1}^d
\frac{\beta\,\boldsymbol{\theta}_l}{\Gamma(1-\beta)}
\int_0^{\infty}\left(f(\mathbf{x})-f(\mathbf{x}-s\boldsymbol{\theta}_l)\right)s^{-\beta-1}ds.
\end{equation}
Our specialization of \eqref{M} provides useful and manageable tools
to treat fractional equations in multidimensional spaces in order to
find explicit solutions. We notice that each vector in the
orthonormal basis
$(\boldsymbol{\theta}_1,.......,\boldsymbol{\theta}_d)$ can be
expressed in terms of the canonical basis $\mathbf{e}_i$ by applying
a rotation matrix, such that
$$\boldsymbol{\theta}_i = \sum_{k = 1}^{d}\theta_{ik} \mathbf{e}_k.$$
The Fourier transform of \eqref{ngr} is given by
\begin{equation}\label{fou}
\widehat{\nabla^{\beta}_{\theta}f}(\mathbf{k})= \sum_{l=1}^d
\boldsymbol{\theta}_l(-i\mathbf{k}\cdot
\boldsymbol{\theta}_l)^{\beta}\widehat{f}(\mathbf{k}).
\end{equation}
A relevant point to understand the consequence of this definition in
the framework of fractional vector calculus is given by the
definition of fractional Laplacian. For $\beta \in(1,2]$, given a
scalar function $f(\mathbf{x})$, with $\mathbf{x}\in \mathbb{R}^d$,
the fractional directional operator corresponding to the definition
\eqref{ngr} is given by
\begin{equation}
\mathbb{D}_{\theta}^{\beta}f(\mathbf{x})=\nabla_{\theta}\cdot
\nabla_{\theta}^{\beta-1}f(\mathbf{x}),
\end{equation}
that is the inverse Fourier transform of
\begin{align}\label{lap3}
\widehat{\mathbb{D}_{\theta}^{\beta}f}(\mathbf{k})= \sum_{l=1}^d
(-i\mathbf{k}\cdot
\boldsymbol{\theta}_l)^{\beta}\widehat{f}(\mathbf{k}).
\end{align}
We remark that the fractional operator \eqref{lap3} is given by the
sum of fractional directional derivatives of order $\beta \in
(1,2]$. Indeed, by inverting \eqref{lap3}, we get
$$\mathbb{D}_{\theta}^{\beta}f(\mathbf{x})= \sum_{l= 1}^d (\boldsymbol{\theta}_l\cdot \nabla)^{\beta}f(\mathbf{x}).$$
In the same way we can give a definition of fractional divergence of
a vector field as follows
\begin{equation}
div^{\beta}\mathbf{u}(\mathbf{x},t)=\nabla^\beta_{\theta}\cdot
\mathbf{u}=\sum_{l=1}^d (\boldsymbol{\theta}_l\cdot
\nabla)^{\beta}\boldsymbol{\theta}_l\cdot \mathbf{u}(\mathbf{x},t),
\end{equation}
with $\beta\in(0,1)$.

\begin{example}

Let us consider the case $\mathbf{x}\in \mathbb{R}^2$. In this case we denote $\boldsymbol{\theta}_1\equiv(\cos\theta_1, \sin \theta_1)$ and
$\boldsymbol{\theta}_2\equiv(\cos\theta_2, \sin \theta_2)$. By definition, these two vectors must be orthonormal, hence $\theta_2= \theta_1+\frac{\pi}{2}$.
These two fixed directions are given by a rotation of the cartesian
axes. In this case the fractional gradient is given by
\begin{equation}
\nabla^{\beta}_{\theta}f(\mathbf{x})\equiv\left[(\cos\theta_1, \sin\theta_1)(\cos\theta_1 \partial_x+\sin\theta_1 \partial_y)^{\beta}+
(\cos\theta_2, \sin\theta_2)(\cos\theta_2 \partial_x+\sin\theta_2 \partial_y)^{\beta}\right]f(\mathbf{x}).
\end{equation}
An interesting discussion about this two-dimensional case can be found in \cite{Ervin}.

\end{example}

\begin{os}

We observe that in the case $\boldsymbol{\theta}_i \equiv
\mathbf{e}_i$, we have the definition of fractional
gradient given by Tarasov. The divergence of this operator brings to the analog of the fractional
Laplacian, given by
\begin{equation}\label{yac}
\nabla_{\theta}\cdot \nabla_{\theta}^{\beta}f(\mathbf{x})=
\sum_{k = 1}^d \frac{\partial}{\partial
x_k}\frac{\partial^{\beta}}{\partial
x_k^{\beta}}f(\mathbf{x}),
\end{equation}
which means that, for $\beta = 1$, we recover the classical
definition of Laplacian. On the other hand, it is well known that in
some cases the Riemann-Liouville derivative does not satisfy
the law of exponent,
$$\frac{\partial}{\partial x}\frac{\partial^{\beta}}{\partial x^{\beta}}f(\mathbf{x})\neq
\frac{\partial^{1+\beta}}{\partial x^{1+\beta}} f(\mathbf{x}).$$
Hence, in this case the fractional heat equation, for $d = 2$, has the
following form
\begin{equation}\nonumber
\frac{\partial}{\partial t}f(x,y,t)=\left(\frac{\partial}{\partial x}\frac{\partial^{\beta}}{\partial x^{\beta}}+\frac{\partial}{\partial y}\frac{\partial^{\beta}}{\partial y^{\beta}}\right)f(x,y,t)
\end{equation}
i.e, a multidimensional heat equation with fractional sequential
derivatives. We observe that \eqref{yac} leads to the Riemann-Liouville fractional analog
of the Laplace operator recently studied by Dalla Riva and Yakubovich in \cite{yacu}.
The physical and probabilistic meaning of this
formulation will be discussed below in relation to the general
formulation concerning Definition 2.2.

\end{os}

\section{Multidimensional fractional directional advection equation}

We study the $d$-dimensional
fractional advection equation by following the approach to fractional vector
calculus suggested in the previous section. From a physical point of view we get
inspiration from \cite{m3}, where the fractional vector calculus has been applied in order to study the flow
of contaminants in an heterogeneous porous medium. First of all we
give a different, original derivation of the fractional
multidimensional advection equation, starting from the continuity
equation, that is
\begin{equation}\label{fic}
\frac{\partial\rho_{\alpha}}{\partial t}= - div^{\alpha}\mathbf{V}, \; \alpha\in(0,1),
\end{equation}
where $\mathbf{V}(\mathbf{x},t)$ is the flux of contaminant particles, that is the vector rate at which mass is transported through a unit surface.
The physical meaning of this fractional conservation of mass can be directly related to the recent paper by Wheatcraft and Meerschaert (\cite{mass}).
The relation between flux and density of contaminants is given by the classical Fick's law, its form in absence of dispersion is simply
\begin{equation}
\mathbf{V}(\mathbf{x},t)= \mathbf{u}\rho_{\alpha}(\mathbf{x},t),
\end{equation}
where $\mathbf{u}$ is the velocity field of contaminant particles; for semplicity in the following discussion we take this velocity field constant in all directions.
By substitution we find the $n$-dimensional fractional advection equation in the following form
\begin{equation}
\frac{\partial\rho_{\alpha}}{\partial t}= -
div^{\alpha}(\mathbf{u}\rho_{\alpha})=-
\nabla_{\theta}^{\alpha}\cdot(\mathbf{u} \rho_{\alpha}).
\end{equation}
Hereafter we denote with $\chi_D$ the characteristic function of the set $D$.\\

We are now ready to state the following

\begin{te}

Let us consider the $d$-dimensional fractional advection equation
\begin{equation}\label{main}
\frac{\partial}{\partial t}\rho_{\alpha} +
\nabla^{\alpha}_{\theta}\cdot (\mathbf{u}\rho_{\alpha})= 0, \quad
\mathbf{x} \in \mathbb{R}^d, \, t>0,
\end{equation}
where $\alpha \in (0,1)$, and $\mathbf{u}\equiv (u_1,....., u_d)$ is
the velocity field, with $u_i$, $i= 1,...,d$, constants. The
solution to \eqref{main}, subject to the initial condition
$$\rho_{\alpha}(\mathbf{x}, 0) = f(\mathbf{x})\in L^1(\mathbb{R}^d),$$
is written as
\begin{equation}\label{oggi}
\rho_{\alpha}(\mathbf{x}, t) =
\int_{\mathbb{R}^d}f(\mathbf{y})\prod_{l=1}^{d}\mathcal{U}_{\alpha}(\boldsymbol{\theta}_l\cdot(\mathbf{x}-\mathbf{y}),
(\mathbf{u}\cdot\boldsymbol{\theta}_l)t)\chi_{\{\boldsymbol{\theta}_l\cdot(\mathbf{x}-\mathbf{y})\geq
0\}}(\mathbf{y})d\mathbf{y},
\end{equation}
 where
$\mathcal{U}_{\alpha}$ is the solution to
\begin{equation}\label{main2}
\left( \frac{\partial}{\partial t} +
\lambda\frac{\partial^{\alpha}}{\partial
x^{\alpha}}\right)\mathcal{U}_{\alpha}(x,t) = 0, \quad x \in
\mathbb{R}_{+}, \, t>0 , \lambda\in \mathbb{R}_{+},
\end{equation}
with initial condition $\mathcal{U}_{\alpha}(x,0)=\delta(x)$.

\end{te}

\begin{proof}

We start by taking the Fourier transform of equation \eqref{main},
given by
\begin{equation}
\frac{\partial}{\partial t}\widehat{\rho_{\alpha}}(\mathbf{k}, t)
+\mathbf{u}\cdot
\widehat{\nabla_{\theta}^{\alpha}\rho_{\alpha}}(\mathbf{k}, t) = 0.
\end{equation}
From \eqref{fou}, we obtain that
\begin{equation}
\left( \frac{\partial}{\partial t} + \left(\sum_{l=1}^d (\mathbf{u}\cdot\boldsymbol{\theta}_l)(-i\mathbf{k}\cdot \boldsymbol{\theta}_l)^{\alpha}\right)\right)
\widehat{\rho_{\alpha}}(\mathbf{k}, t) = 0,
\end{equation}
and by integration we find
\begin{align}\label{main1a}
\widehat{\rho_{\alpha}}(\mathbf{k}, t) &= \widehat{f}(\mathbf{k})exp\left(-t\sum_{l=1}^d (\mathbf{u}\cdot\boldsymbol{\theta}_l)(-i\mathbf{k}\cdot \boldsymbol{\theta}_l)^{\alpha}\right)\\
\nonumber &=\widehat{f}(\mathbf{k}) \prod_{l=1}^d
exp\left(-t(\mathbf{u}\cdot\boldsymbol{\theta}_l)(-i\mathbf{k}\cdot
\boldsymbol{\theta}_l)^{\alpha}\right).
\end{align}
If we take the Fourier transform of equation \eqref{main2}, then we
obtain
\begin{equation}\label{man}
\left( \frac{\partial}{\partial t} +
\lambda(-i\gamma)^{\alpha}\right)\widehat{\mathcal{U}_{\alpha}}(\gamma,t)
= 0,
\end{equation}
where we used the fact that
$$\widehat{\frac{\partial^{\alpha}}{\partial x^{\alpha}}f}(\gamma)=(-i\gamma)^{\alpha}\widehat{f}(\gamma).$$
By integrating \eqref{man}, and by taking into account the initial
condition, we obtain
\begin{equation}
\widehat{\mathcal{U}_{\alpha}}(\gamma,t)= exp(-\lambda t(-i\gamma)^{\alpha}).
\end{equation}
Thus, we can rearrange \eqref{main1a} as follows
\begin{align}
\widehat{\rho}_{\alpha}(\mathbf{k},t)&=
\widehat{f}(\mathbf{k})\prod_{l=1}^d
exp\big(-t(\mathbf{u}\cdot\boldsymbol{\theta}_l)(-i\mathbf{k}\cdot
\boldsymbol{\theta}_l)^{\alpha}\big)\\
\nonumber&=\widehat{f}(\mathbf{k}) \prod_{l=1}^d
\widehat{\mathcal{U}_{\alpha}}(\gamma_l,\lambda_l t)|_{\gamma_l=
\mathbf{k}\cdot\boldsymbol{\theta}_l,\lambda_l=\mathbf{u}\cdot\boldsymbol{\theta}_l}.
\end{align}
Finally, we observe that the inverse Fourier transform of any
\begin{equation}\nonumber
\widehat{\mathcal{U}_{\alpha}}(\mathbf{k}\cdot\boldsymbol{\theta}_l,\lambda_l
t), \quad l = 1,2,\cdots, d,
\end{equation}
is given by
\begin{equation}
{\mathcal{U}_{\alpha}}(\mathbf{x}\cdot\boldsymbol{\theta}_l,\lambda_l
t)\chi_{\{(\mathbf{x}\cdot \boldsymbol{\theta}_l)\geq 0\}}\quad l =
1,2,\cdots, d,
\end{equation}
 and therefore, we get that
\begin{equation}\label{gre}
\rho_{\alpha}(\mathbf{x}, t)=(f\ast G)(\mathbf{x},t),
\end{equation}
where the symbol $\ast$ stands for Fourier convolution, and
\begin{equation}\label{gre}
G(\mathbf{x},t)=\prod_{l=1}^d
\mathcal{U}_{\alpha}(\mathbf{x}\cdot\boldsymbol{\theta}_l,
(\mathbf{u}\cdot\boldsymbol{\theta}_l)t)\chi_{\{(\mathbf{x}\cdot
\boldsymbol{\theta}_l)\geq 0\}}.
\end{equation}
Formula \eqref{gre} can be explicitly written as
\begin{equation}\label{gre1}
\rho_{\alpha}(\mathbf{x,t})=\int_{\mathbb{R}^d}f(\mathbf{y})
G(\mathbf{x}-\mathbf{y},t)d\mathbf{y},
\end{equation}
therefore \eqref{gre1} coincides with \eqref{oggi} and the proof is
completed.
\end{proof}

\bigskip

Let us consider the L\'evy process
$\left(\boldsymbol{X}_t\right)_{t\geq 0}$, with infinitesimal
generator $\mathcal{A}$ and transition semigroup $P_t=
e^{t\mathcal{A}}$. The transition law of
$\left(\boldsymbol{X}_t\right)_{t\geq 0}$ is written as
\begin{equation}\nonumber
P_t u_0(\mathbf{x})=\mathbb{E}u_0(\mathbf{X}_t+\mathbf{x}),
\end{equation}
and solves the Cauchy problem
\begin{equation}\label{levo}
\begin{cases}
\frac{\partial}{\partial t}u(\mathbf{x},t)= (\mathcal{A}u)(\mathbf{x},t),\\
u(\mathbf{x},0)= u_0(\mathbf{x}).
\end{cases}
\end{equation}
We say that the process $\left(\boldsymbol{X}_t\right)_{t\geq 0}$ is the stochastic solution of \eqref{levo}.
We also consider the integral representation of $\mathcal{A}$, given by
\begin{equation}\label{af}
\mathcal{A} f(\mathbf{x})=
\frac{1}{(2\pi)^d}\int_{\mathbb{R}^d}e^{-i\mathbf{k}\cdot\mathbf{x}}\Phi(\mathbf{k})\widehat{f}(\mathbf{k})d\mathbf{k},
\end{equation}
for all functions $f$ in the domain
\begin{equation}
D(\mathcal{A})=\big\{f(\mathbf{x})\in L^1_{loc}(\mathbb{R}^d, d
\mathbf{x}):\int_{\mathbb{R}^d}\Phi(\mathbf{k})|\widehat{f}(\mathbf{k})|^2d\mathbf{k}<\infty\big\}
\end{equation}
Then, we say that $P_t$ is a pseudo-differential operator with
symbol $\widehat{P}_t= exp(t\Phi)$ and $\Phi$ is the Fourier multiplier of
$\mathcal{A}$. Furthermore from the characteristic function of the
process $(\boldsymbol{X}_t)_{t\geq 0}$, we obtain that
\begin{equation}
\left[\frac{\partial}{\partial t}\mathbb{E}e^{i\mathbf{k}\cdot
\boldsymbol{X}_t}\right]_{t=0}= \Phi(\mathbf{k}).
\end{equation}
We also recall that a stable subordinator
$(\mathfrak{H}^\alpha_t)_{t> 0}$, $\alpha \in (0,1)$, is a L\'evy process with
non-negative, independent and stationary increments, whose law, say
$h_{\alpha}(x,t)$, $x\geq 0$, $t\geq 0$, has the Laplace transform
\begin{equation}\label{lap}
\tilde{h}_{\alpha}(s,t)=\int_0^{+\infty}e^{-sx}h_{\alpha}(x,t)dx=e^{-t
s^{\alpha}}, s\geq 0.
\end{equation}
For more details on this topic we refer to \cite{Bertoin}.\\

Let $P_t$ be the semigroup associated with \eqref{main}, then, for all $t>0$
\begin{equation}\label{bound0}
\|P_t f\|_{\infty}\leq d\|f\|_{L^1}.
\end{equation}
Indeed, from the fact that
$$\|\mathcal{U}_{\alpha}(\cdot, t)\|_{\infty}\leq 1, \mbox{\; uniformly}$$
and, from \eqref{gre},
$$\|G(\cdot, t)\|_{\infty}\leq d\|\mathcal{U}_{\alpha}(\cdot, t)\|_{\infty},$$
we have that
\begin{equation}\nonumber
\|P_t f\|_{\infty}\leq d\|\mathcal{U}_{\alpha}(\cdot, t)\|_{\infty}\|f\|_{L^1}\leq d\|f\|_{L^1}.
\end{equation}
We present the following result concerning the equation
\eqref{main}.

\begin{te}

The stochastic solution to the $d$-dimensional fractional advection
equation \eqref{main}, subject to the initial condition
$\rho_{\alpha}(\mathbf{x}, 0)= \delta(\mathbf{x})$, is given by the
process
$$\boldsymbol{Z}_t = \sum_{l=1}^d \boldsymbol{\theta}_l\mathfrak{H}^{\alpha}_l(\lambda_l t),\quad t\geq 0,$$
which is a random vector in $\mathbb{R}^d$, where for $l = 1,...,d$,
$\lambda_l = \mathbf{u}\cdot\boldsymbol{\theta}_l$ and
$\mathfrak{H}^{\alpha}_l(t)$, $t>0$, are independent $\alpha$-stable
subordinators.

\end{te}

\begin{proof}

We recall that
\begin{equation}\label{c1}
\widehat{\rho_{\alpha}}(\mathbf{k}, t) = \prod_{l=1}^d
exp\big((-t(\mathbf{u}\cdot\boldsymbol{\theta}_l)(-i\mathbf{k}\cdot
\boldsymbol{\theta}_l)^{\alpha}\big),
\end{equation}
is the Fourier transform of the solution to \eqref{main}, with
initial condition $\rho_0(\mathbf{x})=\delta(\mathbf{x})$. By using
\eqref{lap}, formula \eqref{c1} can be written as
\begin{align}\label{carat}
\widehat{\rho_{\alpha}}(\mathbf{k}, t) &= \prod_{l=1}^d \mathbb{E}
exp\big((i\mathbf{k}\cdot
\boldsymbol{\theta}_l)\mathfrak{H}^\alpha_l(\lambda_l t)\big)\\
 \nonumber & =\mathbb{E} exp\left(i\sum_{l=1}^d (\mathbf{k}\cdot
\boldsymbol{\theta}_l)\mathfrak{H}^\alpha_l(\lambda_l t)\right)\\
\nonumber &=\mathbb{E} exp\left(i \mathbf{k}\cdot\sum_{l=1}^d
\boldsymbol{\theta}_l\mathfrak{H}^\alpha_l(\lambda_l t)\right)=
\mathbb{E}e^{i\mathbf{k}\cdot\boldsymbol{Z}_t}.
\end{align}
Hence $\rho_{\alpha}$ is the law of the process $\boldsymbol{Z}_t = \sum_{l=1}^d \boldsymbol{\theta}_l \mathfrak{H}^\alpha_l(\lambda_l t)$,
that is a random vector whose components are given by different linear combination of $d$ independent $\alpha$-stable subordinators.

\end{proof}
We observe that these processes can be studied in the general framework of L\'evy additive processes.\\

We now study the Cauchy problem for the multidimensional fractional advection equation with random initial data.
The theory of random solutions of partial differential equations
has a long history, starting from the pioneeristic works of Kamp\'e de F\'eriet (\cite{kam}).

\begin{te}
Let us consider the Cauchy problem
\begin{equation}\label{ranc}
\begin{cases}
\frac{\partial}{\partial t}\rho_{\alpha} +
\nabla^{\alpha}_{\theta}\cdot (\mathbf{u}\rho_{\alpha})= 0, \quad
\mathbf{x} \in \mathbb{R}^d, \, t>0,\, \alpha\in(0,1),\\
\rho_{\alpha}(\mathbf{x},0)=X(\mathbf{x})\in L^2(\mathbb{R}),
\end{cases}
\end{equation}
where the random field $X(\mathbf{x})$, $\mathbf{x}\in \mathbb{R}_+^d$, is a random initial condition
 $X:\left(\Omega, \mathcal{A}, P\right)\mapsto
\left(\mathbb{R},B(\mathbb{R}),
e^{-x^2/2}/\sqrt{2\pi}\right),$ such
that
\begin{equation}\label{condi1}
X(\mathbf{x})=\sum_{j\in \mathbb{N}}c_j \varphi_j(\mathbf{x}), \quad c_j= \int_{\mathbb{R}^d}X(\mathbf{x})\varphi_j(\mathbf{x})d\mathbf{x},
\end{equation}
where $\{\Phi_j\}$ is dense in $L^2(\mathbb{R})$.
Then, the stochastic solution of \eqref{ranc} is given by
\begin{equation}
\rho_{\alpha}(\mathbf{x},t)=\sum_{j\in\mathbb{N}}c_j P_t \varphi_j(\mathbf{x}),
\end{equation}
where $P_t$ is the transition semigroup associated with \eqref{main}.

\end{te}

\begin{proof}
Since $X\in L^2$, then there exists an orthonormal system $\{\varphi_j:j\in \mathbb{N}\}$ such that
\eqref{condi1} holds true in $L^2$. Indeed the first identity in \eqref{condi1} must be understood in
$L^2(dP\times d\mathbf{x})$ sense as follows
\begin{equation}
\lim_{L\rightarrow \infty}\mathbb{E}\left[\int_{\mathbb{R}^d}\left(X(\mathbf{x})-\sum_{j=0}^L c_j \varphi_j(\mathbf{x})\right)^2
d\mathbf{x}\right]=0
\end{equation}
From Theorem 3.2, we know that $\mathbf{Z}_t$ is the stochastic solution to the $d$-dimensional
fractional advection equation \eqref{main}. In view
of these facts we write the solution of \eqref{ranc} as follows
\begin{align}\label{ramo}
\rho_{\alpha}(\mathbf{x},t)&=\mathbb{E}\left[X(\mathbf{x}+\mathbf{Z}_t)|\mathcal{F}_X\right]\\
\nonumber &= \mathbb{E}\left[\sum_{j\in \mathbb{N}} c_j\varphi_j(\mathbf{x}+\mathbf{Z}_t)|\mathcal{F}_X\right]\\
\nonumber &= \sum_{j\in \mathbb{N}} c_j\mathbb{E}\varphi_j(\mathbf{x}+\mathbf{Z}_t),
\end{align}
where $\mathcal{F}_X$ is the $\sigma$-algebra generated by $X$ and we recall that
$$c_j= \int_{\mathbb{R}^d}X(\mathbf{x})\varphi_j(\mathbf{x})d\mathbf{x}.$$
We observe that
$$\mathbb{E}\varphi_j(\mathbf{x}+\mathbf{Z}_t)= P_t \varphi_j(\mathbf{x}),$$
is the solution to the Cauchy problem
\begin{equation}
\begin{cases}\label{ranc1}
\frac{\partial}{\partial t}\rho_{\alpha} +
\nabla^{\alpha}_{\theta}\cdot (\mathbf{u}\rho_{\alpha})= 0, \quad
\mathbf{x} \in \mathbb{R}^d_{+}, \, t>0,\\
\rho_{\alpha}(\mathbf{x},0)=\varphi_j(\mathbf{x}).
\end{cases}
\end{equation}
Therefore, \eqref{ramo} becomes
\begin{equation}
\rho_{\alpha}(\mathbf{x},t)= \sum_{j\in\mathbb{N}}c_j P_t \varphi_j(\mathbf{x}),
\end{equation}
and solves \eqref{ranc} as claimed, being \eqref{ranc1} satisfied term by term.
Also, from the fact that $P_0= Id$, we get that
$$\rho_{\alpha}(\mathbf{x},0)= \sum_{j\in\mathbb{N}}c_j P_0 \varphi_j(\mathbf{x})= \sum_{j\in\mathbb{N}}c_j \varphi_j(\mathbf{x})=X(\mathbf{x}).$$
If $X$ is represented as \eqref{condi1}, then $X$ is square-summable, that is
$$\int_{\mathbb{R}^d}X^2(\mathbf{x})d\mathbf{x}=\sum_{j\in\mathbb{N}}c_j^2< \infty.$$
Therefore, from \eqref{bound0}, we have that
$$\|\rho_{\alpha}(\cdot,t)\|_{\infty}\leq \sum_{j \in \mathbb{N}}|c_j|  \|\varphi_j\|_{\infty}< \infty.$$

\end{proof}

\subsection{Multidimensional fractional advection-dispersion equation}

We follow our approach to study a general fractional
advection-dispersion equation (FADE). We provide a
multidimensional nonlocal formulation of the Fick's law, written as
follows

\begin{equation}\label{liga}
\mathbf{V}(\mathbf{x},t)=
-\nu\nabla_{\theta}^{\beta-1}\rho_{\beta}(\mathbf{x},t), \quad \beta
\in (1,2), \nu\in\mathbb{R^+},
\end{equation}
such that
\begin{equation}
\nabla \cdot\mathbf{V}(\mathbf{x},t)
=-\nu\mathbb{D}_{\theta}^{\beta}\rho_{\beta}(\mathbf{x},t)
\end{equation}
The one-dimensional fractional Fick's law has been at the core
of many recent papers (see for example \cite{Para} and the references therein).
The total flux in the conservation of mass \eqref{fic} is given by the sum of the advective flux  and the dispersive
flux. Hence we obtain
the formulation of the FADE investigated in the next theorem.

\begin{te}

Let us consider the $d$-dimensional fractional advection-dispersion
equation
\begin{equation}\label{fade}
 \frac{\partial}{\partial t}\rho_{\alpha,\beta}+  \nabla^{\alpha}_{\theta}\cdot(\mathbf{u}\rho_{\alpha,\beta}) =
  \mathbb{D}_{\theta}^{\beta}\rho_{\alpha,\beta}, \quad \mathbf{x} \in \mathbb{R}^d, \,
  t>0,
\end{equation}
where $\alpha \in (0,1)$, $\beta \in (1,2)$ and $\mathbf{u}\equiv
(u_1,....., u_n)$ is the velocity field, with $u_i$, $i= 1,...d$,
are constants. The solution to \eqref{fade}, subject to the initial
condition
$$\rho_{\alpha,\beta}(\mathbf{x}, 0) = \delta(\mathbf{x}),$$
is written as
$$\rho_{\alpha,\beta}(\mathbf{x}, t)=\prod_{l=1}^d \mathcal{U}_{\alpha}(\boldsymbol{\theta}_l
\cdot \mathbf{x},
(\mathbf{u}\cdot\boldsymbol{\theta}_l)t)\ast\mathcal{U}_{\beta}(\boldsymbol{\theta}_l\cdot
\mathbf{x},t)\chi_{(\mathbf{x}\cdot \boldsymbol{\theta}_l)\geq 0},$$
where $\ast$ stands for convolution with respect to $\mathbf{x}$,
$\mathcal{U}_{\alpha}$ is the solution to the one-dimensional
fractional advection equation
\begin{equation}
\left( \frac{\partial}{\partial t} +
\lambda\frac{\partial^{\alpha}}{\partial
x^{\alpha}}\right)\mathcal{U}_{\alpha}(x,t) = 0, \quad x \in
\mathbb{R}_{+}, \, t>0 , \lambda\in \mathbb{R}_{+},
\end{equation}
with initial condition $\mathcal{U}_{\alpha}(x,0)=\delta(x)$ and
$\mathcal{U}_{\beta}$ is the solution of the space-fractional
diffusion equation
\begin{equation}\label{dif}
\left( \frac{\partial}{\partial t} -
\frac{\partial^{\beta}}{\partial
x^{\beta}}\right)\mathcal{U}_{\beta}(x,t) = 0, \quad x \in
\mathbb{R}_{+}, \, t>0, \beta\in (1,2).
\end{equation}

\end{te}

\begin{proof}

The proof follows the same arguments of Theorem 3.1.
To begin with, we take the Fourier transform of equation
\eqref{fade}: by using \eqref{fou}, we obtain
\begin{equation}
\left(\frac{\partial}{\partial t} + (\sum_{l=1}^d
(\mathbf{u}\cdot\boldsymbol{\theta}_l)(-i\mathbf{k}\cdot
\boldsymbol{\theta}_l)^{\alpha})\right)
\widehat{\rho_{\alpha,\beta}}(\mathbf{k}, t) = \left(\sum_{l=1}^d
(-i\mathbf{k}\cdot \boldsymbol{\theta}_l)^{\beta}\right)
\widehat{\rho_{\alpha,\beta}}(\mathbf{k}, t),
\end{equation}
and by integration we find
\begin{align}\label{main1}
\widehat{\rho_{\alpha,\beta}}(\mathbf{k}, t) &=
exp\left(-t\sum_{l=1}^d
(\mathbf{u}\cdot\boldsymbol{\theta}_l)(-i\mathbf{k}\cdot
\boldsymbol{\theta}_l)^{\alpha}\right)
exp\left(t\sum_{l=1}^d (-i\mathbf{k}\cdot \boldsymbol{\theta}_l)^{\beta}\right)\\
\nonumber &=\prod_{l=1}^d
exp\left(-t(\mathbf{u}\cdot\boldsymbol{\theta}_l)(-i\mathbf{k}\cdot
\boldsymbol{\theta}_l)^{\alpha}\right)exp\left(t(-i\mathbf{k}\cdot
\boldsymbol{\theta}_l)^{\beta}\right).
\end{align}
On the other hand if we take the Fourier transform of equation
\eqref{dif}, we obtain
\begin{equation}
\left( \frac{\partial}{\partial t} -
(-i\gamma)^{\beta}\right)\widehat{\mathcal{U}_{\beta}}(\gamma,t) =
0, \quad \beta \in (1,2),
\end{equation}
then, integrating, we obtain
\begin{equation}
\widehat{\mathcal{U}_{\beta}}(\gamma,t)= exp( t(-i\gamma)^{\beta}).
\end{equation}
Thus, we can rearrange \eqref{fade} in the following
way
\begin{align}
\widehat{\rho}_{\alpha,\beta}(\mathbf{k},t)&=\prod_{l=1}^d\bigg(
\widehat{\mathcal{U}_{\alpha}}(\gamma_l,\lambda_l
t)\bigg|_{\gamma_l= \mathbf{k}
\cdot\boldsymbol{\theta}_l,\lambda_l=\mathbf{u}\cdot\boldsymbol{\theta}_l}\bigg)\bigg(\widehat{\mathcal{U}_{\beta}}(\gamma_l,t)\bigg|_{\gamma_l=
\mathbf{k}\cdot\boldsymbol{\theta}_l}\bigg)\\
\nonumber &=\prod_{l=1}^d
exp\left(-t(\mathbf{u}\cdot\boldsymbol{\theta}_l)(-i\mathbf{k}\cdot
\boldsymbol{\theta}_l)^{\alpha}\right)exp\left(t(-i\mathbf{k}\cdot
\boldsymbol{\theta}_l)^{\beta}\right).
\end{align}
Finally, from the convolution theorem, we conclude the proof.
\end{proof}
For the reader's convenience, we recall that the explicit form of
the fundamental solution of the Riemann-Liouville space-fractional
equation \eqref{dif} can be found for example in \cite{Main}. It is also
possible to give an explicit form to the solution of \eqref{fade} in
terms of
one-sided stable probability density function.\\
We notice that in \eqref{liga} we have considered two different
order $\alpha \neq \beta$, respectively for the advection and
dispersion term. Indeed, from a physical point of view the two orders
$\alpha$ and $\beta$ can be different, although they are certainly
related. The parameter $\alpha$ was introduced from the fractional
conservation of mass, hence it depends by the geometry of the porous
medium. The parameter $\beta$ takes into account nonlocal effects in
the Fick's law. Both of them are physically related to the
heterogeneity of the porous medium; an explicit relation between
them must be object of further investigations.

\begin{os}
The stochastic solution to \eqref{fade} is given by the sum of a random vector whose components are given by different linear combination of
$d$ independent $\alpha$-stable subordinators ($Z_t$ in Theorem 3.2) and a multivariable $\alpha$-stable random vector with discrete spectral measure. This second term corresponds
to the unique case in which an $\alpha$-stable random vector has independent components (see \cite{Taqqu}).
The proof is a direct consequence of Theorem 2.1 and 3.2.
\end{os}

\section{Fractional power of operators and fractional shift operator}

    In order to highlight the applications of the fractional gradient, we recall some general results about fractional power of operators.
    The final aim is to find an operational rule for a shift operator involving fractional gradients, in analogy  with the exponential shift operator.
    A power $\alpha$ of a closed linear operator $\mathcal{A}$ can be represented by means of the Dunford integral (\cite{Komatsu})
    \begin{equation}\label{AalphaDunf}
    \mathcal{A}^\alpha = \frac{1}{2\pi i} \int_\Gamma d\lambda \, \lambda^\alpha \, (\lambda - \mathcal{A})^{-1}, \quad \Re\{ \alpha \} >0
    \end{equation}
    under the conditions
    \begin{equation*}
    \begin{array}{rl} (i) & \lambda \in \rho(\mathcal{A}) \, (\textrm{the resolvent set of  } \mathcal{A}) \textrm{ for all } \lambda >0;\\
    (ii) & \| \lambda (\lambda I + \mathcal{A})^{-1} \| < M < \infty \textrm{ for all }\lambda >0  \end{array}
    \end{equation*}
    where $\Gamma$ encircles the spectrum $\sigma(\mathcal{A})$ counterclockwise avoiding the negative real axis and $\lambda^\alpha$ takes the principal branch.
    For $\Re\{\alpha\} \in (0,1)$, the integral \eqref{AalphaDunf} can be rewritten in the Bochner sense as follows
    \begin{equation}\label{AalphaBoch}
    \mathcal{A}^\alpha = \frac{\sin \pi \alpha}{\pi} \int_0^\infty d\lambda\, \lambda^{\alpha-1} (\lambda + \mathcal{A})^{-1} \mathcal{A}.
    \end{equation}
    By inserting (Hille-Yosida theorem)
    \begin{equation}\nonumber
    (\lambda + \mathcal{A})^{-1} = \int_0^\infty dt \, e^{-\lambda t} e^{-t \mathcal{A}}
    \end{equation}
    into \eqref{AalphaBoch} we get that
    \begin{align*}
    \int_0^\infty d\lambda\, \lambda^{\alpha-1} (\lambda + \mathcal{A})^{-1} = \left( \int_0^\infty s^{-\alpha} e^{-s} ds \right)
    \left( \int_0^\infty ds\, s^{\alpha -1} e^{-s\mathcal{A}} \right)
    \end{align*}
    where
    \begin{equation*}
    \int_0^\infty s^{-\alpha} e^{-s} ds = \Gamma(1-\alpha), \quad \alpha \in (0,1)
    \end{equation*}
    and
    \begin{equation*}
    \frac{1}{\Gamma(\alpha)} \int_0^\infty ds\, s^{\alpha -1} e^{-s\mathcal{A}} = \mathcal{A}^{\alpha -1}
    \end{equation*}
    which holds only if $0 < \alpha < 1$.
    The representation \eqref{AalphaBoch} can be therefore rewritten as
    \begin{equation*}
    \mathcal{A}^\alpha = \mathcal{A}^{\alpha - 1} \mathcal{A}, \quad \alpha \in (0,1).
    \end{equation*}
    and, for $\alpha \in (0,1)$, we get that
    \begin{equation}
    \mathcal{A}^\alpha = \mathcal{A} \mathcal{A}^{\alpha -1} = \mathcal{A} \left[ \frac{1}{\Gamma(1-\alpha)} \int_0^\infty ds\, s^{-\alpha}
    e^{-s \mathcal{A}} \right].
    \end{equation}
    On the other hand we can write the fractional power of the operator $\mathcal{A}$ as follows
    \begin{equation}
    \mathcal{A}^{\alpha} = \mathcal{A}^n \mathcal{A}^{\alpha-n}, \quad n-1 < \alpha < n, \; n \in \mathbb{N},
    \end{equation}
    and therefore, we can immediately recover the Riemann-Liouville fractional derivative of order
    $\alpha \in (0,1)$ as a fractional power of the ordinary first derivative $\mathcal{A} = \partial_x$
    (see for example \cite{libro}).\\
    We also remark that, given the operator $\mathcal{A}$ as before, the strong solution to the space fractional equation
    $$ \left( \frac{\partial}{\partial t} + \mathcal{A}^{\alpha} \right) u(x,t) = 0$$
    subject to a good initial condition $u(x, 0)=u_0(x)$,
    can be represented as the convolution
    \begin{equation} \label{convAalpha}
    u(x, t) = e^{-t \mathcal{A}^{\alpha}} u_0(x) = \mathbb{E} e^{- \mathfrak{H}^{\alpha}_t \, \mathcal{A}} \, u_0(x),
    \end{equation}
    in the sense that
    $$\lim_{t\rightarrow 0}\Bigg\|\frac{e^{-t\mathcal{A}^{\alpha}}u-u}{t}-\mathcal{A}^{\alpha}u\Bigg\|_{L^p(\mu)}=0,$$
    for some $p\geq 1$, with a Radon measure $\mu$.
    In \eqref{convAalpha}, we recall that
    $\mathfrak{H}^{\alpha}_t$, with $t>0$, is the $\alpha$-stable subordinator and
    \begin{equation}\label{koma}
    \mathbb{E} e^{- \mathfrak{H}^{\alpha}_t \, \mathcal{A}} =  \int_0^\infty ds\, h_{\alpha}(s, t) \, e^{-s \mathcal{A}},
    \end{equation}
    where $h_{\alpha}$ is the density law of the stable subordinator.
    For $\alpha=1$, we obtain the solution
    $$u(x,t) = e^{-t\mathcal{A}}u_0(x),$$
    from the fact that, we formally have that
    $$\lim_{\alpha \to 1}h_\alpha(x, s) = \delta(x-s). $$
    Indeed, for $\alpha \to 1$, we get that $\mathfrak{H}^\alpha_t \stackrel{a.s.}{\longrightarrow} t$ which is the elementary subordinator (\cite{Bertoin}).
    Equation \eqref{convAalpha} appears of interest in relation to operatorial methods in quantum mechanics and, generally
    to solve differential equations.
    Actually, we recall the notion of exponential shift operator.
    It is well known that
    \begin{equation}
    e^{\theta \partial_x}f(x) = f(x+\theta), \theta\in \mathbb{R},
    \end{equation}
    for $f(x)\in C_b(0,+\infty)$, that is the space of continuous bounded functions (\cite{Hille}). This operational rule comes directly
    from the Taylor expansion of the analytic function $f(x)$ near $x$.
    It provides a clear physical meaning to this operator as a generator of translations in quantum mechanics.\\
    In a recent paper, Miskinis (\cite{Miskinis}) discusses the properties of the generalized one-dimensional quantum  operator of the momentum in the framework
    of the fractional quantum mechanics. This is a relevant topic because of the role of the momentum operator as a generator of translation.
    In its analysis he suggested the following definition of the generealized momentum
    \begin{equation}\label{Misk}
    \hat{p}= C\frac{\partial^\alpha}{\partial x^{\alpha}}, \;\alpha \in (0,1),
    \end{equation}
    with $C$ a complex coefficient, such that, if $\alpha = 1$ then we have the classical quantum operator $\hat{p}= -i\hbar\partial_x$.\\
    In the same way, under the previous analysis  we can introduce a fractional shift operator as
    \begin{equation}\label{con1}
    e^{-\theta \partial_x^\alpha}f(x) = \int_0^\infty ds\, h_\alpha(s, \theta) \,e^{-s\partial_x}
    f(x)= \int_0^\infty ds\, h_\alpha(s, \theta) \,
    f(x-s), \quad \theta >0.
    \end{equation}
    This fractional operator does not give a pure translation, it is a convolution of the initial condition with the density law of the stable subordinator,
    stressing again the possible role of this stochastic analysis in the framework of the fractional quantum mechanics. However, in the special case $\alpha = 1$,
    it gives again the classical shift operator. This operational rule has a direct interpretation in relation to the definition of a generalized quantum
    operator, similar to that of \eqref{Misk}. This stochastic view of the generator of translations can be, in our view, a good starting point
    for further investigations. Moreover, we can generalize these considerations to multidimensional fractional operators and give the operational
    solution of a general class of fractional equations as follows

\begin{prop}
    Consider the multidimensional fractional advection equation
    \begin{equation}
    \left(\frac{\partial}{\partial t}+ \sum_{i=1}^d\frac{\partial^{\alpha}}{\partial x_i^{\alpha}}\right)
    \rho_{\alpha}(\mathbf{x},t)=0, \; \alpha \in (0,1),\; \mathbf{x}\in \mathbb{R}^d_+, t>0,
    \end{equation}
    subject to the initial and boundary conditions
    $$\rho_{\alpha}(\mathbf{x}, 0) = \prod_{i=1}^d\rho_{0}(x_i), \qquad \rho_{\alpha}(\boldsymbol{0}, t)=0.$$
    Then its analytic solution is given by
    \begin{equation}
    \rho_{\alpha}(\mathbf{x}, t) = e^{-t\sum_{i=1}^d \partial_{x_i}^{\alpha}}\rho_{\alpha}(\mathbf{x}, 0).
    \end{equation}
\end{prop}

\begin{proof}
We can write
\begin{align}
\rho_{\alpha}(\mathbf{x},t)&= e^{-t\sum_{i=1}^d\partial_x^{\alpha}}\rho_{0}(\mathbf{x},0)\\
\nonumber &= \prod_{i=1}^d
e^{-t\partial_{x_i}^{\alpha}}\rho_{0}(x_i,0).
\end{align}
Hence, by direct application of \eqref{con1} we have
\begin{align}
\rho_{\alpha}(\mathbf{x},t)&= \prod_{i=1}^d \int_0^\infty ds\, h_\alpha(s, t) \, \rho_{0}(x_i-s)\\
\nonumber &=\int_0^\infty ds\, h_\alpha(s, t) \,
\rho_{0}(\mathbf{x}-s).
\end{align}
Thus, we conclude that
\begin{equation}
\rho_{\alpha}(\mathbf{x}, t)=\int_0^{\infty}ds h_{\alpha}(s,t)\rho_0(\mathbf{x}-s) =
e^{-t\sum_{i=1}^d \partial_{x_i}^{\alpha}}\rho_{\alpha}(\mathbf{x}, 0),
\end{equation}
as claimed.

\end{proof}

Let us recall definition and main properties of the compound
Poisson process. Consider a sequence of
independent $\mathbb{R}^n$-valued random variables $Y_i$, $i\in \mathbb{N}$, with identical law $\nu(\cdot)$.
Let $(N_t)_{t\geq 0}$ be a Poisson
process with intensity $\lambda >0$. The compound Poisson process is
the L\'evy process
\begin{equation}
X_t= \sum_{i=1}^{N(t)}\tau(Y_i),
\end{equation}
with infinitesimal generator (see for example \cite{Jacobs}, pag.131)
\begin{equation}
\mathcal{A}f(x)=\int_{\mathbb{R}^n}\left(f(x+\tau(y))-f(x)\right)\nu(dy).
\end{equation}
We state the following result about the stochastic processes driven by equations involving the fractional gradient \eqref{ngr}.\\

\begin{te}
Let us consider the random vector $(\boldsymbol{Z}_t)_{t\geq 0}$ in
$\mathbb{R}^d$, given by
\begin{equation}\label{comPu}
\boldsymbol{Z}_t= \sum_{j =1}^{d}\boldsymbol{\theta}_j
\mathfrak{H}_j^\alpha(X_t),
\end{equation}
where $\mathfrak{H}^\alpha_j$ are independent $\alpha$-stable
subordinators, with $\alpha \in (0,1)$ and $(X_t)_{t\geq 0}$ is an independent
compound Poisson process $$X_t= \sum_{i=1}^{N(t)}\tau(Y_i),$$ with
$\tau:\mathbb{R}^d\mapsto \mathbb{R}_+$.
The infinitesimal generator
of the process \eqref{comPu} is given by
\begin{equation}
(\mathcal{A}f)(\mathbf{x})= \sum_{j=1}^d\int_{\mathbb{R}^d}\left[(e^{-\tau(\mathbf{y})(\boldsymbol{\theta}_j\cdot\nabla)^{\alpha}}-1)f(\mathbf{x})\right]\nu(d\mathbf{y}).
\end{equation}
Moreover assuming that $\boldsymbol{\theta}_j\equiv \mathbf{e}_j$, $\forall j \in \mathbb{N}$ and
$$f(\mathbf{x}) = \prod_{i = 0}^d g_i(x_i),$$
where $g_i(x_i)$ are analytic functions, we find
\begin{equation}\label{seco}
(\mathcal{A}f)(\mathbf{x})=
\sum_{j=1}^d\int_{\mathbb{R}^d}\left\{\left[\int_0^{+\infty}ds\,
h_{\alpha}(s,\tau(\mathbf{y}))g_j(x_j-s)\right]-g_j(x_j)\right\}\nu(d\mathbf{y}),
\end{equation}
where $h_{\alpha}$ is the density law of a stable subordinator.
\end{te}

\begin{proof}

The characteristic function of the random vector \eqref{comPu} is
\begin{align}
\mathbb{E}e^{i\mathbf{k}\cdot\boldsymbol{Z}_t}&=\mathbb{E}exp\left(i\sum_{j=1}^{d}\mathbf{k}\cdot \boldsymbol{\theta}_j
\mathfrak{H}_j^\alpha(X_t)\right)\\
\nonumber &=\prod_{j=1}^d \mathbb{E}exp\left(i\mathbf{k}\cdot
\boldsymbol{\theta}_j \mathfrak{H}_j^\alpha(X_t)\right)\\
\nonumber &= \prod_{j=1}^d \mathbb{E}exp\left(-X_t(-i\mathbf{k}\cdot\boldsymbol{\theta}_j)^{\alpha} \right)\\
\nonumber &=  \prod_{j=1}^d exp\left(-\lambda
t(1-\mathbb{E}e^{-(-i\mathbf{k}\cdot\boldsymbol{\theta}_j)^{\alpha}\tau(Y)})\right).
\end{align}

Then by differentiation we can find the Fourier multiplier
\begin{align}
\Phi(\mathbf{k})&=\left[\partial_t \mathbb{E}e^{i\mathbf{k}\cdot\boldsymbol{Z}_t}\right]_{t=0}\\
\nonumber &=\lambda\sum_{j=1}^d
\int_{\mathbb{R}^d}\left(e^{-(-i\mathbf{k}\cdot\boldsymbol{\theta}_j)^{\alpha}\tau(\mathbf{y})}-1\right)\nu(d\mathbf{y}),
\end{align}
of the generator $\mathcal{A}$,
where $\nu(\cdot)$ is the law of the jumps of the compound Poisson process.
Finally we have, by inverse Fourier transform
\begin{equation}
(\mathcal{A}f)(\mathbf{x})=
\sum_{j=1}^d\int_{\mathbb{R}^d}\left[(e^{-\tau(\mathbf{y})(\boldsymbol{\theta}_j\cdot
\nabla)^{\alpha}}-1)f(\mathbf{x})\right]\nu(d\mathbf{y}).
\end{equation}
In order to prove \eqref{seco}, we notice that
\begin{equation}
e^{-t\partial_x^{\alpha}}f(x)=\mathbb{E}e^{-\mathfrak{H}_t^{\alpha}\partial_x}f(x),
\end{equation}
where
$$\mathbb{E}e^{-\mathfrak{H}_t^{\alpha}\partial_x}=\int_0^{+\infty}ds \, h_{\alpha}(s,t)e^{-s\partial_x},$$
and $h_{\alpha}$ is the density law of the stable
subordinator.
Recalling that, given an analytic function, the exponential operator
acts as a shift operator, i.e.
$$e^{-t\partial_x}f(x)= f(x-t),$$
we find that
\begin{equation}
e^{-\tau(\mathbf{y})\partial_{x_i}^{\alpha}}f(x_i)=\int_0^{+\infty}ds\,
h_{\alpha}(s,\tau(\mathbf{y}))f(x_i-s)ds.
\end{equation}
Hence, assuming that
$$f(\mathbf{x})=\prod_{k=1}^d g_i(x_i),$$
in the case $\boldsymbol{\theta}_j \equiv \mathbf{e}_j$, $\forall j \in \mathbb{N}$, we conclude that
\begin{equation}
(\mathcal{A}f)(\mathbf{x})=
\sum_{j=1}^d\int_{\mathbb{R}^d}\left\{\left[\int_0^{+\infty}ds\,
h_{\alpha}(s,\tau(\mathbf{y}))g_j(x_j-s)\right]-g_j(x_j)\right\}\nu(d\mathbf{y}).
\end{equation}

\end{proof}

\section{L\'evy-Khinchine formula with fractional gradient}

In this section we discuss some results about Markov processes related to the above definition of fractional gradient.
We present a new version of the L\'evy-Khinchine formula involving fractional operators and we discuss some possible applications.
It is well known that the L\'evy-Khinchine formula provides a representation of characteristic functions of
infinitely divisible distributions. Let us recall that, given a one-dimensional L\'evy process $(X_t)_{t\geq 0}$, we have
\begin{align}
\mathbb{E}e^{ikX_t}&=e^{\Phi(k)t},
\end{align}
with characteristic exponent given by
\begin{equation}
\Phi(k)=ik b-\frac{k^2 c}{2}+\int_{\mathbb{R}}\left(e^{ikx}-1-(ikx)\chi_{\{|x|<1\}}\right)\nu(dx),
\end{equation}
where $b\in \mathbb{R}$ is the drift term, $c\in \mathbb{R}$ is the diffusion term and $\nu(\cdot)$ is a L\'evy measure.

In the following we will consider the case $b = c = 0$. In this case the infinitesimal generator of $(X_t)_{t\geq 0}$, is given by
\begin{equation}\label{lk}
\mathcal{A}f (x)=\frac{1}{(2\pi)}\int_{\mathbb{R}}e^{-ik x}
\Phi(k)\widehat{f}(k)dk =\int_{\mathbb{R}}\left(f(x+y)-f(x)- y \partial_x f(x)\,\chi_{\{|y|<1\}}\right)\nu(dy).
\end{equation}
Hereafter the symbol "$\sim$" stands for equality in law.

\begin{te}

Consider the random vector $(\boldsymbol{Z}_t)_{t\geq 0}$ in
$\mathbb{R}^d$, given by
\begin{equation}\label{comP}
\boldsymbol{Z}_t= \sum_{j
=1}^{N(t)}\mathbf{Y}_j-\sum_{l=1}^d\boldsymbol{\theta}_l(\boldsymbol{\theta}_l\cdot\mathbb{E}\mathbf{Y})^{1/\alpha}
\mathfrak{H}^\alpha_l(\lambda t)\chi_{\mathcal{D}}(\mathbf{Y}),
\end{equation}
where
$$\mathcal{D}=\{\mathbf{Y}\in \mathbb{R}^d: \mathbb{E}(\boldsymbol{\theta}_l\cdot\mathbf{Y})>0, l = 1, \cdots, d\},$$
$\mathfrak{H}^\alpha_l$ are i.i.d $\alpha$-stable
subordinators, with $\alpha \in (0,1)$, and $\mathbf{Y}_j$ are
$d$-dimensional i.i.d. random vectors such that $\mathbf{Y}_j\sim
\mathbf{Y}$, for all $j\in \mathbb{N}$ and
$P(\mathbf{Y} \in A )=\int_A \nu(d\mathbf{y}),$
as before.
Then, the infinitesimal generator of the process \eqref{comP} is
given by
\begin{equation}\label{lk}
\mathcal{L}^{\theta}f(\mathbf{x})=
\int_{\mathbb{R}^d}\left[(f(\mathbf{x}+\mathbf{y})-f(\mathbf{x})-
\mathbf{y}\cdot\boldsymbol{\nabla}_{\theta}^{\alpha}f(\mathbf{x})\chi_{D(\boldsymbol{\theta})}(\mathbf{y})\right]\nu(d\mathbf{y}),
\end{equation}
where $\boldsymbol{\nabla}_{\theta}^{\alpha}$ is the fractional
gradient in the sense of equation \eqref{ngr}, and
$$D(\boldsymbol{\theta})=\bigcap_{l=1}^d\{\mathbf{y}\in\mathbb{R}^d: \boldsymbol{\theta}_l\cdot \mathbf{y}\geq
0\}.$$

\end{te}

\begin{proof}
We consider the characteristic function of the random vector
\eqref{comP}
\begin{equation}\label{equo}
\mathbb{E}e^{i\mathbf{k}\cdot\boldsymbol{Z}_t}=\mathbb{E}exp\left(i\sum_{j=1}^{N(t)}\mathbf{Y}_j\cdot\mathbf{k}\right)\,\mathbb{E}exp\left(i\sum_{l=1}^d
(\boldsymbol{\theta}_l\cdot\mathbb{E}\mathbf{Y})^{1/\alpha}\mathfrak{H}^\alpha_l(\lambda
t)(\mathbf{k}\cdot\boldsymbol{\theta}_l) \chi_{\mathcal{D}}(\mathbf{Y})\right).
\end{equation}
The first term can be written as follows
\begin{equation}
\mathbb{E}exp\left(i\sum_{j=1}^{N(t)}\mathbf{Y}_j\cdot\mathbf{k}\right)
=\mathbb{E}\left(\mathbb{E}e^{i\sum_{j=1}^{N(t)}\mathbf{Y}_j\cdot\mathbf{k}}\right),
\end{equation}
and, from the fact that $\mathbf{Y_j}\sim\mathbf{Y}$, we have
\begin{align}\label{compa}
\mathbb{E}\left(\mathbb{E}e^{i\sum_{j=1}^{n}\mathbf{Y}_j\cdot\mathbf{k}}|N(t)=n\right)&=
\mathbb{E}\left((\mathbb{E}e^{i\mathbf{Y}\cdot\mathbf{k}})^{n}|N(t)=n\right)\\
\nonumber &=
\sum_{n=0}^{\infty}[\mathbb{E}e^{i\mathbf{Y}\cdot\mathbf{k}}]^n
Pr\{N(t)=n\}\\
\nonumber &=\sum_{n=0}^{\infty}[\mathbb{E}e^{i\mathbf{Y}\cdot\mathbf{k}}]^n\frac{(\lambda t)^n}{n!}e^{-\lambda t}\\
\nonumber &= e^{-\lambda
t(1-\mathbb{E}e^{i\boldsymbol{Y}\cdot\mathbf{k}})}.
\end{align}
Regarding the second term in \eqref{equo}, we notice that
\begin{equation}
(\boldsymbol{\theta}_l\cdot\mathbb{E}\mathbf{Y})^{1/\alpha}
\mathfrak{H}^\alpha_l(\lambda
t)\chi_{\mathcal{D}}(\mathbf{Y})\stackrel{d}{=}\mathfrak{H}^\alpha_l((\boldsymbol{\theta}_l\cdot\mathbb{E}\mathbf{Y})\lambda
t)\chi_{\mathcal{D}}(\mathbf{Y}),
\end{equation}
where $\stackrel{d}{=}$ stands for equality in distribution.
From the fact that $\mathfrak{H}^\alpha_l$ are i.i.d $\alpha$-stable
subordinators, we obtain
\begin{align}
\mathbb{E}exp\left(i\sum_{l=1}^d(\boldsymbol{\theta}_l\cdot\mathbb{E}\mathbf{Y})^{1/\alpha}\mathfrak{H}^\alpha_l(\lambda
t)\mathbf{k}\cdot\boldsymbol{\theta}_l \chi_{\mathcal{D}}(\mathbf{Y})\right)&=\prod_{l=1}^d\mathbb{E}
exp\left(i(\boldsymbol{\theta}_l\cdot\mathbb{E}\mathbf{Y})^{1/\alpha}\mathfrak{H}^\alpha_l(\lambda
t)(\mathbf{k}\cdot\boldsymbol{\theta}_l) \chi_{\mathcal{D}}(\mathbf{Y})\right)\\
\nonumber &=\prod_{l=1}^d exp\left(-\lambda
t(-i\mathbf{k}\cdot\boldsymbol{\theta}_l
)^{\alpha}(\boldsymbol{\theta}_l\cdot \mathbb{E}\mathbf{Y}) \chi_{\mathcal{D}}(\mathbf{Y})\right)
\end{align}
Finally, we get
\begin{equation}
\mathbb{E}e^{i\mathbf{k}\cdot\boldsymbol{Z}_t}= exp\left( \lambda t
(\mathbb{E}e^{i\mathbf{k}\cdot\mathbf{Y}}-1-\sum_{l=1}^d(-i\mathbf{k}\cdot\boldsymbol{\theta}_l)^{\alpha}(\boldsymbol{\theta}_l\cdot\mathbb{E}\mathbf{Y})
\chi_{\mathcal{D}}(\mathbf{Y}).
\right),
\end{equation}
where the Fourier multiplier $-\Phi(\mathbf{k})$, of $\mathcal{L}^{\theta}$, is
given by
\begin{align}
\Phi(\mathbf{k})&=\left[\partial_t \mathbb{E}e^{i\mathbf{k}\cdot\boldsymbol{Z}_t}\right]_{t=0}\\
\nonumber &=\lambda\left(\mathbb{E}e^{i\mathbf{k}\cdot\mathbf{Y}}-1-\sum_{l=1}^d(-i\mathbf{k}\cdot\boldsymbol{\theta}_l)^{\alpha}(\boldsymbol{\theta}_l\cdot\mathbb{E}\mathbf{Y})\chi_{\mathcal{D}}(\mathbf{Y})\right)\\
\nonumber &= \lambda\int_{\mathbb{R}^d}\left[e^{i\mathbf{k}\cdot\mathbf{y}}-1-\left(\sum_{l=1}^d(-i\mathbf{k}\cdot\boldsymbol{\theta}_l)^{\alpha}(\boldsymbol{\theta}_l\cdot
\mathbf{y})\right)\chi_{D(\theta)}(\mathbf{y})\right]\nu(d\mathbf{y}),
\end{align}
where we recall that $\nu(\cdot)$ is the law of $Y$.
Then, we can
use equation \eqref{af} and, by inverse Fourier transform, we get
\begin{equation}
\mathcal{L}^{\theta}f(\mathbf{x})=\int_{\mathbb{R}^d}\left[(f(\mathbf{x}+\mathbf{y})-
f(\mathbf{x}))-
\mathbf{y}\cdot\boldsymbol{\nabla}_{\theta}^{\alpha}f(\mathbf{x})\chi_{D(\theta)}(\mathbf{y})\right]\nu(d\mathbf{y}),
\end{equation}
which is the claim.
\end{proof}

\begin{os}
In the case $\boldsymbol{\theta}_l \equiv \mathbf{e}_l$, for all $l$,
$$D(\boldsymbol{\theta})=\bigcap_{l=1}^d\{\mathbf{y}\in\mathbb{R}^d: \boldsymbol{e}_l\cdot \mathbf{y}\geq
0\}\equiv \mathbb{R}^d_{+},$$
and \eqref{lk} becomes
\begin{equation}
\mathcal{L}^{\theta}f(\mathbf{x})=
\int_{\mathbb{R}^d}\left[(f(\mathbf{x}+\mathbf{y})-f(\mathbf{x})-
\sum_{j=1}^d y_j\partial_{x_j}^{\alpha}f(\mathbf{x})\chi_{\mathbb{R}^{+}}(\mathbf{y})\right]\nu(d\mathbf{y}).
\end{equation}

\end{os}

\begin{os}
We observe that in the special case $d = 1$, $\alpha = 1$, the
process \eqref{comP} becomes the compensated Poisson process
\begin{equation}
Z_t= \sum_{j =1}^{N(t)}Y_j-\lambda t\mathbb{E}Y, \quad t>0.
\end{equation}
In this case, the law of $(Z_t)_{t\geq 0}$ is given by
\begin{equation}
P(Z_t\in dy)/dy= \sum_{n=0}^{\infty}f_Y^{*n}(y+\lambda
t\mathbb{E}Y)e^{-\lambda t}\frac{(\lambda t)^n}{n!},
\end{equation}
where $f_Y$ is the law of the jumps $Y_j\sim Y$ and $f^{*n}$ is the $n$-convolution of $f_Y$.
Straightforward calculations lead to the explicit representation of the law
for $\alpha \neq 1$. Indeed for $\alpha \in (0,1)$, we have that
\begin{equation}
P(Z_t\in dy)/dy= \sum_{n=0}^{\infty}\mathbb{E}f_Y^{*n}(y+(\lambda\mathbb{E}Y)^{1/\alpha}\mathfrak{H}^\alpha_t).
\end{equation}
\end{os}

Let us consider the random vector $\mathbf{W}$, whose components are
independent folded Gaussian random variables with variance
$rE_{\beta}$, where $rE_{\beta}$ is the inverse Gamma distribution,
with probability density function given by
$$P\{rE_{\beta}\in ds\}/ds= \frac{1}{\Gamma(\beta)}\left(\frac{s}{r}\right)^{-\beta-1}e^{-\frac{r}{s}}, \quad s\geq 0$$
where $r\geq 0$ is a scale parameter and $\beta>0$ a shape
parameter. We observe that
\begin{align}\label{mis}
P\{\mathbf{W}\in d\mathbf{y}\}/d\mathbf{y}&=
2^d\int_0^{\infty}\frac{e^{-\frac{|\mathbf{y}|^2}{4s}}}{\sqrt{(4\pi
s)^d}}P\{rE_{\beta}\in ds\}\\
\nonumber &= \frac{2^d
r^{\beta}}{\sqrt{(4\pi)^d}\Gamma(\beta)}\int_0^{\infty}s^{-\beta-1-\frac{d}{2}}e^{-s^{-1}(\frac{|\mathbf{y}|^2}{4}+r)}ds\\
\nonumber
&=\frac{\Gamma(\beta+\frac{d}{2})}{\Gamma(\beta)}\frac{2^{2(\beta+d)}
}{\sqrt{(4\pi)^d}}\frac{r^{\beta}}{\left(|\mathbf{y}|^2+4r\right)^{\beta+\frac{d}{2}}}=m_r(|\mathbf{y}|^2).
\end{align}
Then we have that
\begin{align}
\mathbb{E}W_j&=\frac{1}{\Gamma(\beta)}\frac{2
r^{\beta+1}}{\sqrt{4\pi}}\int_0^{\infty}\int_0^{\infty}y\,
s^{-\beta-\frac{3}{2}}e^{-s^{-1}(\frac{\mathbf{y}^2}{4}+r)}ds
dy\\
\nonumber &=\frac{\Gamma(\beta-\frac{1}{2})}{\Gamma(\beta)}\frac{2
r^{\frac{3}{2}}}{\sqrt{\pi}}.
\end{align}
We assume that the random vectors $\mathbf{Y}_j$ appearing in
\eqref{comP} are taken such that
\begin{equation}\label{jump}
\mathbf{Y}_j\sim \epsilon_j \mathbf{W}_j, j\in \mathbb{N},
\end{equation}
where $\epsilon_j$ is the Rademacher random variable, i.e.
$P(\epsilon_j = +1)=p$ and $P(\epsilon_j = -1)=q$ and $\mathbf{W}_j$
are the i.i.d random vectors distributed like $\mathbf{W}$.
It is worth to notice that, in this case, the set $\mathcal{D}$ is given by
$$\mathcal{D}=\{(p-q)\boldsymbol{\theta}_l\cdot\mathbb{E}\mathbf{W})>0, l = 1, \cdots, d\},$$
where $\mathbb{E}\mathbf{W}$ is positive.\\
We are now able to state the following theorem.

\begin{te}
Let us consider the process \eqref{comP} with jumps \eqref{jump}. For $p\neq q$ and $\beta \in(0,1/2)$, we have that
\begin{equation}
 \boldsymbol{Z}(t/r^{\beta})\xrightarrow[r\rightarrow 0]{d}\boldsymbol{Q}(t),
\end{equation}
where $\boldsymbol{Q}(t)$, $t\geq 0$, has generator
\begin{align}\label{scal}
\mathcal{L}_{p,q}^{\theta}f(\mathbf{x})&=
C_d(\beta)\int_{\mathbb{R}^d}\left[(p\, f(\mathbf{x}+\mathbf{y})+q\,
f(\mathbf{x}-\mathbf{y})-f(\mathbf{x})- (p-q)
\mathbf{y}\cdot\boldsymbol{\nabla}_{\theta}^{\alpha}f(\mathbf{x})\chi_{D(\boldsymbol{\theta})}(\mathbf{y})\right]
\frac{d\mathbf{y}}{|\mathbf{y}|^{2\beta+d}}\\
\label{scal1}&=C_d(\beta)p\int_{\mathbb{R}^d}\left[(
f(\mathbf{x}+\mathbf{y})-f(\mathbf{x})-
\mathbf{y}\cdot\boldsymbol{\nabla}_{\theta}^{\alpha}f(\mathbf{x})\chi_{D(\boldsymbol{\theta})}(\mathbf{y})\right]
\frac{d\mathbf{y}}{|\mathbf{y}|^{2\beta+d}}\\
\label{q} &+C_d(\beta) q\int_{\mathbb{R}^d}\left[(
f(\mathbf{x}-\mathbf{y})-f(\mathbf{x})+
\mathbf{y}\cdot\boldsymbol{\nabla}_{\theta}^{\alpha}f(\mathbf{x})\chi_{D(\boldsymbol{\theta})}(\mathbf{y})\right]
\frac{d\mathbf{y}}{|\mathbf{y}|^{2\beta+d}},
\end{align}
with $p,q\geq 0$ such that $p+q=1$.
\end{te}

\begin{proof}

Under the assumption that $\mathbf{Y}_j\sim \epsilon_j \mathbf{W}_j$
in \eqref{comP}  we have that
\begin{align}
\mathbb{E}\mathbf{Y}=p\mathbb{E}\mathbf{W}-q\mathbb{E}\mathbf{W}=(p-q)\mathbb{E}\mathbf{W}.
\end{align}
Hence, we have
\begin{equation}
\mathbf{Z}_t =\sum_{j
=1}^{N(t)}\epsilon_j\mathbf{W}_j-\sum_{l=1}^d\boldsymbol{\theta}_l((p-q)\boldsymbol{\theta}_l\cdot\mathbb{E}\mathbf{W})^{1/\alpha}
\mathfrak{H}^\alpha_l(\lambda t)\chi_{\mathcal{D}}(\epsilon\mathbf{W}).
\end{equation}
Its characteristic function is given by
\begin{equation}\label{equo1}
\mathbb{E}e^{i\mathbf{k}\cdot\boldsymbol{Z}_t}=\mathbb{E}exp\left(i\sum_{j=1}^{N(t)}\epsilon_j\mathbf{W}_j\cdot\mathbf{k}\right)\,
\mathbb{E}exp\left(i\sum_{l=1}^d
((p-q)\boldsymbol{\theta}_l\cdot\mathbb{E}\mathbf{W})^{1/\alpha}\mathfrak{H}^\alpha_l(\lambda
t)\mathbf{k}\cdot\boldsymbol{\theta}_l\, \chi_{\mathcal{D}}(\epsilon\mathbf{W})\right).
\end{equation}
The first operand in \eqref{equo1} can be written as follows
\begin{align}
\mathbb{E}exp\left(i\sum_{j=1}^{N(t)}\epsilon_j\mathbf{W}_j\cdot\mathbf{k}\right)&=\mathbb{E}\left(\mathbb{E}e^{i\sum_{j=1}^{n}\epsilon_j\mathbf{Y}_j\cdot\mathbf{k}}|N(t)=n\right)=
e^{-\lambda
t(1-\mathbb{E}e^{i\epsilon\boldsymbol{Y}\cdot\mathbf{k}})}\\
\nonumber &= e^{-\lambda
t((p+q)-p\mathbb{E}e^{i\boldsymbol{Y}\cdot\mathbf{k}}-q\mathbb{E}e^{-i\boldsymbol{Y}\cdot\mathbf{k}})}.
\end{align}
The second term in \eqref{equo1}, being $\mathfrak{H}^\alpha_l$
i.i.d $\alpha$-stable subordinators, is given by
\begin{align}
&\mathbb{E}exp\left(i\sum_{l=1}^d((p-q)\boldsymbol{\theta}_l\cdot\mathbb{E}\mathbf{W})^{1/\alpha}\mathfrak{H}^\alpha_l(\lambda
t)(\mathbf{k}\cdot\boldsymbol{\theta}_l)\,\chi_{\mathcal{D}}(\epsilon \mathbf{W})\right)\\
\nonumber &=\prod_{l=1}^d\mathbb{E}exp\left(i((p-q)\boldsymbol{\theta}_l\cdot\mathbb{E}\mathbf{W})^{1/\alpha}\mathfrak{H}^\alpha_l(\lambda
t)\mathbf{k}\cdot\boldsymbol{\theta}_l\,\chi_{\mathcal{D}}(\epsilon\mathbf{W})\right)\\
\nonumber &=\prod_{l=1}^d exp\left(-\lambda
t(-i\mathbf{k}\cdot\boldsymbol{\theta}_l
)^{\alpha}((p-q)\boldsymbol{\theta}_l\cdot
\mathbb{E}\mathbf{W})\chi_{\mathcal{D}}(\epsilon \mathbf{W})\right).
\end{align}
Finally, we have that
\begin{align}
\mathbb{E}e^{i\mathbf{k}\cdot\boldsymbol{Z}_t}&= exp\left( \lambda t
(p\mathbb{E}e^{i\mathbf{k}\cdot\mathbf{Y}}+q
\mathbb{E}e^{-i\mathbf{k}\cdot\mathbf{Y}}-1)-\sum_{l=1}^d(-i\mathbf{k}\cdot\boldsymbol{\theta}_l)^{\alpha}((p-q)\boldsymbol{\theta}_l\cdot\mathbb{E}\mathbf{W})
\chi_{\mathcal{D}}(\epsilon \mathbf{W})\right)\\
\nonumber &=e^{t\Phi_r(\mathbf{k})},
\end{align}
where
\begin{equation}
\Phi_r(\mathbf{k})=\lambda\int_{\mathbb{R}^d}[pe^{i\mathbf{k}\cdot\mathbf{y}}+qe^{-i\mathbf{k}\cdot\mathbf{y}}-1-(p-q)\sum_{l=1}^d(-i\mathbf{k}\cdot\boldsymbol{\theta}_l)^{\alpha}(\boldsymbol{\theta}_l\cdot
\mathbf{y})\chi_{D(\theta)}(\mathbf{y})]m_r(|\mathbf{y}|^2),
\end{equation}
with $m_r(|\mathbf{y}|^2)$ given by equation \eqref{mis}.\\

We now
consider the process $\boldsymbol{Z}(t/r^{\beta})$, whose
characteristic function is given by
\begin{equation}
\mathbb{E}e^{i\mathbf{k}\cdot\boldsymbol{Z}(t/r^{\beta})}= exp\left(
\frac{t}{r^{\beta}}\Phi_r(\mathbf{k})\right).
\end{equation}
Then, we get the Fourier symbol

\begin{align}
&\left[\partial_t \mathbb{E}e^{i\mathbf{k}\cdot\boldsymbol{Z}(t/r^{\beta})}\right]_{t=0}=\frac{1}{r^{\beta}}\Phi_r(\mathbf{k}) \\
\nonumber
&=\frac{\lambda}{r^{\beta}}\left(p\mathbb{E}e^{i\mathbf{k}\cdot\mathbf{Y}}+q\mathbb{E}e^{-i\mathbf{k}\cdot\mathbf{Y}}-1-\sum_{l=1}^d(-i\mathbf{k}\cdot\boldsymbol{\theta}_l)^{\alpha}
((p-q)\boldsymbol{\theta}_l\cdot\mathbb{E}\mathbf{Y}\chi_{\mathcal{D}}(\epsilon \mathbf{W})\right)\\
\nonumber &=C_d(\beta)\lambda\int_{\mathbb{R}^d}\left[p
e^{i\mathbf{k}\cdot\mathbf{y}}+q
e^{-i\mathbf{k}\cdot\mathbf{y}}-1-(p-q)\sum_{l=1}^d(-i\mathbf{k}\cdot\boldsymbol{\theta}_l)^{\alpha}(\boldsymbol{\theta}_l\cdot
\mathbf{y})\chi_{D(\theta)}(\mathbf{y})\right]\frac{d\mathbf{y}}{\left(|\mathbf{y}|^2+4r\right)^{\beta+\frac{d}{2}}},
\end{align}
where
$$C_d(\beta)=\frac{\Gamma(\beta+\frac{d}{2})}{\Gamma(\beta)}\frac{2^{2(\beta+d)}
}{\sqrt{(4\pi)^d}}.$$
 This implies that the process
$\boldsymbol{Q}(t)$, obtained from
$$\boldsymbol{Z}(t/r^{\beta})\xrightarrow[r\rightarrow
0]{d}\boldsymbol{Q}(t),$$
has a generator with Fourier multiplier
\begin{equation}\label{la}
\frac{1}{r^{\beta}}\Phi_r(\mathbf{k}) \xrightarrow{r\rightarrow
0}\Phi(\mathbf{k}),
\end{equation}
where
\begin{equation}
\Phi(\mathbf{k})= C_d(\beta)\lambda\int_{\mathbb{R}^d}\left[pe^{i\mathbf{k}\cdot\mathbf{y}}+qe^{-i\mathbf{k}\cdot\mathbf{y}}-1-(p-q)\sum_{l=1}^d(-i\mathbf{k}\cdot\boldsymbol{\theta}_l)^{\alpha}(\boldsymbol{\theta}_l\cdot
\mathbf{y})\chi_{D(\theta)}(\mathbf{y})\right]\frac{d\mathbf{y}}{|\mathbf{y}|^{2\beta+ d}}.
\end{equation}

We conclude that the generator of the process $\boldsymbol{Q}(t)$
is given by the inverse Fourier transform of \eqref{la}, i.e.
\begin{align}\label{inte}
&\mathcal{L}_{p,q}^{\theta}f(\mathbf{x})=
C_d(\beta)\int_{\mathbb{R}^d}\left[(p\, f(\mathbf{x}+\mathbf{y})+q\,
f(\mathbf{x}-\mathbf{y})-f(\mathbf{x})- (p-q) \mathbf{y}\cdot
\boldsymbol{\nabla}_{\theta}^{\alpha}f(\mathbf{x})\chi_{D(\boldsymbol{\theta})}(\mathbf{y})\right]\frac{1}{|\mathbf{y}|^{2\beta+d}},
\end{align}
as claimed.\\

We now study the convergence of the integral
\eqref{inte}. By taking the multidimensional MacLaurin expansion of
the integrand up to the second order term, we have
\begin{align}\label{Mac}
&p\, f(\mathbf{x}+\mathbf{y})+q\,
f(\mathbf{x}-\mathbf{y})-(p+q)f(\mathbf{x})- (p-q) \mathbf{y}\cdot
\boldsymbol{\nabla}_{\theta}^{\alpha}f(\mathbf{x})\chi_{D(\boldsymbol{\theta})}(\mathbf{y})\\
&\nonumber \approx (p-q)\left[\mathbf{y}\cdot\nabla f(\mathbf{x})-\mathbf{y}\cdot\nabla^{\alpha}_{\theta} f(\mathbf{x})
\chi_{D(\boldsymbol{\theta})}(\mathbf{y})\right]+ |\mathbf{y}|^2\Delta f(\mathbf{x}),
\end{align}
hence we obtain
\begin{align}
&\frac{|p\, f(\mathbf{x}+\mathbf{y})+q\,
f(\mathbf{x}-\mathbf{y})-(p+q)f(\mathbf{x})- (p-q) \mathbf{y}\cdot
\boldsymbol{\nabla}_{\theta}^{\alpha}f(\mathbf{x})\chi_{D(\boldsymbol{\theta})}(\mathbf{y})|}{|\mathbf{y}|^{2\beta+d}}\\
\nonumber &\leq \frac{|P(|\mathbf{y}|)|\|D^{2,\alpha}_{p,q}f(\mathbf{x})\|_{\infty}}{|\mathbf{y}|^{2\beta+d}},
\end{align}
where $P(z)$ is a second order polynomial in the variable
$z=|\mathbf{y}|$, arising from the MacLaurin expansion \eqref{Mac}.
We observe that
\begin{equation}
|\nabla^{\beta}_{\theta}f|\leq
\int_{\mathbb{R}^d} |\widehat{\nabla^{\beta}_{\theta}f}(\mathbf{x})|d
\mathbf{x},
\end{equation}
and therefore, by definition (see \eqref{fou}),
$$|\nabla^{\beta}_{\theta}f|<+\infty.$$
Due to the first order term appearing in $|P(|\mathbf{y}|)|$, we
have that
\begin{equation}
\frac{|P(|\mathbf{y}|)|\|D^{2,\alpha}_{p,q}f\|_{\infty}}{|\mathbf{y}|^{2\beta+d}}\leq
\frac{\|D^{2,\alpha}_{p,q}f\|_{\infty}}{|\mathbf{y}|^{2\beta+d-1}},
\end{equation}
which implies that \eqref{inte} converges for $\beta \in (0,1/2)$.
The same reasoning applies for the convergence of \eqref{scal1} and \eqref{q}.

\end{proof}

We notice that, considering jumps \eqref{jump}, the second term in \eqref{comP} reduces to a sum of orthonormal vectors, whose
components are given by independent stable subordinators, that is
\begin{equation}
\sum_{l=1}^d\boldsymbol{\theta}_l((p-q)\boldsymbol{\theta}_l\cdot\mathbb{E}\mathbf{W})^{1/\alpha}
\mathfrak{H}^\alpha_l(\lambda t)\chi_{\mathcal{D}}(\epsilon \mathbf{W})=
r^{\frac{3}{2\alpha}}\sum_{l=1}^d\boldsymbol{\theta}_l\, C_l \mathfrak{H}^\alpha_l(\lambda t)\chi_{\mathcal{D}}(\epsilon\mathbf{W}),
\end{equation}
where
$$C_l = \left((p-q)\frac{2\Gamma(\beta-\frac{1}{2})}{\sqrt{\pi}\Gamma(\beta)}\sum_{i=1}^d\theta_{li}\right)^{1/\alpha}.$$

Moreover, we observe that, by considering zero-mean jumps in \eqref{comP}, we obtain that (see for
example \cite{Mirko1})
\begin{equation}
\boldsymbol{Z}_{t/r^{\beta}}=\sum_{j
=1}^{N(t/r^{\beta})}\mathbf{Y}_j\xrightarrow{d}\mathbf{S}^{2\beta}_t, \quad \mbox{as $r\rightarrow 0$,}
\end{equation}
where $(\mathbf{S}_t)_{t\geq 0}$ is an isotropic vector of stable processes.

\begin{os}
We recall that the fractional Laplacian is defined as follows
\begin{align}
(-\Delta)^{\alpha}f(\mathbf{x})&= p.v. \,C_{d}(\alpha)\int_{\mathbb{R}^d}\frac{f(\mathbf{x})-f(\mathbf{y})}{|\mathbf{x}-\mathbf{y}|^{2\alpha+d}}d\mathbf{y}\\
\nonumber & = \frac{C_d(\alpha)}{2}\int_{\mathbb{R}^d}\frac{f(\mathbf{x}+\mathbf{y})+f(\mathbf{x}-\mathbf{y})-2f(\mathbf{x})}{|\mathbf{y}|^{2\alpha+d}}d\mathbf{y},
\end{align}
where $\alpha \in (0,1)$ and $"p.v.''$ stands for "principle value''. Also, the fractional Laplacian is commonly defined in terms of its Fourier transform, i.e.
\begin{equation}\label{lapf}
-(-\Delta)^{\alpha}f(\mathbf{x})=\frac{1}{(2\pi)^d}\int_{\mathbb{R}^d}e^{-i\mathbf{k}\cdot
\mathbf{x}}|\mathbf{k}|^{2\alpha}\widehat{f}(\mathbf{k})d\mathbf{k},
\end{equation}
with domain given by the Sobolev space of $L^2$ functions for which \eqref{lapf} converges.\\
Formula \eqref{scal}, for $p=q=1/2$, takes the form
\begin{equation}\label{plap}
\mathcal{L}_{1/2,1/2}^{\theta}f(\mathbf{x})=
\frac{C_d(\beta)}{2}\int_{\mathbb{R}^d}\left[(\, f(\mathbf{x}+\mathbf{y})+\,
f(\mathbf{x}-\mathbf{y})-
2f(\mathbf{x})\right]\frac{d\mathbf{y}}{|\mathbf{y}|^{2\beta+d}}=-(-\Delta)^{\beta}f(\mathbf{x}),
\end{equation}
which is independent from the direction $\boldsymbol{\theta}$.
We observe that \eqref{plap} converges for $\beta \in (0,1)$.
This comes from the fact that the first order term in $P(|\mathbf{y}|)$ disappears and therefore
\begin{equation}
\frac{|P(|\mathbf{y}|)|\|D^{2,\alpha}_{p,q}f\|_{\infty}}{|\mathbf{y}|^{2\beta+d}}\leq
\frac{\|D^{2,\alpha}_{p,q}f\|_{\infty}}{|\mathbf{y}|^{2\beta+d-2}}.
\end{equation}
Moreover, by a simple change of variable in \eqref{q}, we recover \eqref{scal1}.
Then, we get that
\begin{align}\label{new}
\mathcal{L}_{1/2,1/2}^{\theta}f(\mathbf{x})&=\frac{C_d(\beta)}{2}\int_{\mathbb{R}^d}\left[(
f(\mathbf{x}+\mathbf{y})-f(\mathbf{x})-
\mathbf{y}\cdot\boldsymbol{\nabla}_{\theta}^{\alpha}f(\mathbf{x})\chi_{D(\boldsymbol{\theta})}(\mathbf{y})\right]
\frac{d\mathbf{y}}{|\mathbf{y}|^{2\beta+d}}\\
\nonumber &+ \frac{C_d(\beta)}{2}\int_{\mathbb{R}^d}\left[( f(\mathbf{x}+\mathbf{y})-f(\mathbf{x})-
\mathbf{y}\cdot\boldsymbol{\nabla}_{\theta}^{\alpha}f(\mathbf{x})\chi_{D(\boldsymbol{\theta})}(-\mathbf{y})\right]\frac{d\mathbf{y}}{|\mathbf{y}|^{2\beta+d}}\\
\nonumber &= C_d(\beta)\int_{\mathbb{R}^d}\left[(
f(\mathbf{x}+\mathbf{y})-f(\mathbf{x})-
\frac{\mathbf{y}}{2}\cdot\boldsymbol{\nabla}_{\theta}^{\alpha}f(\mathbf{x})\right]
\frac{d\mathbf{y}}{|\mathbf{y}|^{2\beta+d}},
\end{align}
being
$$\mathbb{R}^d \equiv \{\mathbf{y}\in \mathbb{R}^d: \boldsymbol{\theta}\cdot\mathbf{y}\geq 0\} \cup
\{\mathbf{y}\in \mathbb{R}^d: \boldsymbol{\theta}\cdot\mathbf{y}< 0\}.$$
We also remark that, in this case, equation \eqref{new} gives a new representation of the fractional Laplacian operator involving fractional gradients.
On the other hand, much care must be done about the convergence. As discussed in the previous theorem, the representation \eqref{new} is convergent
for $\beta\in(0,1/2)$.
\end{os}

\begin{os}
We notice that formula \eqref{scal} includes as special cases,
completely positively or negatively skewed operators. Indeed, we have
the specular cases
\begin{equation}
\begin{cases}
&\mathcal{L}_{1,0}^{\theta}f(\mathbf{x})=
C_d(\beta)\int_{\mathbb{R}^d}\left[(\,
f(\mathbf{x}+\mathbf{y})-f(\mathbf{x})- \mathbf{y}\cdot
\boldsymbol{\nabla}_{\theta}^{\alpha}f(\mathbf{x})\chi_{D(\boldsymbol{\theta})}(\mathbf{y})\right]\frac{1}{|\mathbf{y}|^{2\beta+d}},\\
&\mathcal{L}_{0,1}^{\theta}f(\mathbf{x})=
C_d(\beta)\int_{\mathbb{R}^d}\left[(f(\mathbf{x}-\mathbf{y})-f(\mathbf{x})+
\mathbf{y}\cdot
\boldsymbol{\nabla}_{\theta}^{\alpha}f(\mathbf{x})\chi_{D(\boldsymbol{\theta})}(\mathbf{y})\right]\frac{1}{|\mathbf{y}|^{2\beta+d}}.
\end{cases}
\end{equation}
The first operator is the infinitesimal generator governing processes with only positive jumps, the second one with purely negative jumps.
\end{os}

\begin{os}
It is well known that the generator of the subordinate process
$(\mathbf{X}_{\mathfrak{H}^\alpha_t})_{t>0}$ is given by
\begin{equation}
-(-\mathcal{L})^{\alpha}f(\mathbf{x})=\frac{\alpha}{\Gamma(1-\alpha)}\int_0^{\infty}(P_s
f(\mathbf{x})-f(\mathbf{x}))\frac{ds}{s^{\alpha+1}}
\end{equation}
where $P_s = e^{s\mathcal{L}}$ is the Feller semigroup of the L\'evy process
$(\mathbf{X}_s)_{s>0}$ (see for example \cite{libre}).
\end{os}

\section{Frobenius-Perron operator and fractional equations}

In this section we recall some results about
transport equations involving Frobenius-Perron operator. Then we
show some applications of this approach in the framework of
differential equations involving fractional operators. In particular, we
consider the transport equation
\begin{equation}\label{fro}
\frac{\partial u}{\partial t}= Au- \lambda(I-K)u,
\end{equation}
where
\begin{equation}
Au = -\sum_{k=1}^n\frac{\partial}{\partial x_k}(a(\mathbf{x})u),
\end{equation}
and $K$ is the Frobenius-Perron operator associated with the map $T: x\mapsto x-\tau(x)$ (see for example \cite{Traple}). The
stochastic solution, say $(\mathbf{X}_t)_{t\geq 0}$, to \eqref{fro} is the solution to the stochastic differential equation
$$d\mathbf{X}_t= a(\mathbf{X}_t)dt+\tau(\mathbf{X}_t)d\mathbf{N}_t,$$
where $(\mathbf{N}_t)_{t\geq 0}$ is the Poisson process such that
\begin{equation}
d\mathbf{N}_t=
\begin{cases}
1, \quad \mbox{Poisson arrival at time $t$},\\
0, \quad \mbox{otherwise}.
\end{cases}
\end{equation}

We notice that, if $a(\mathbf{x})=0$,
then $K$ is the backward operator $B$ and $u(k,t)= p_k(t)$, $k\in \mathbb{N}$, $t>0$, becomes the law of the
homogeneous Poisson process. Indeed formula \eqref{fro} takes the form
\begin{align}
\partial_t p_k(t)&=-\lambda(I-B)p_k(t)\\
\nonumber &=- \lambda (p_k(t)-p_{k-1}(t)).
\end{align}
On the other hand, as already pointed out before, the compound Poisson process
$$Z_t = \sum_{j=1}^{N(t)}Y_j,$$
has a generator written as
\begin{equation}
\mathcal{A}f (x)=\int_{\mathbb{R}}\left(f(x-y)-f(x)\right)P(Y\in
dy),
\end{equation}
where the jump $\tau$ equals $Y$ with law $P(Y\in dy)/dy$.

We can now state the following

\begin{te}

Let us consider the process
\begin{equation}\label{per5}
Z_t = \sum_{j=1}^{N(t)}Y_j,
\end{equation}
with
\begin{equation}
Y_j \stackrel{d}{=} Y(1), \quad \forall j \in \mathbb{N},
\end{equation}
where $(Y_t)_{t\geq 0}$ is the stochastic process driven by
\begin{equation}
\frac{\partial f}{\partial t}=\mathcal{G}f.
\end{equation}
Then \eqref{per5} is the stochastic solution to the equation

\begin{equation}\label{per}
\frac{\partial u}{\partial t}= -\lambda(I-e^{\mathcal{G}})u, \quad x\in \mathbb{R}, t>0.
\end{equation}

\end{te}

This means that the law of the jumps in the compound Poisson process
is fixed by the operator $\mathcal{G}$.

\begin{proof}

The process $(Y_t)_{t\geq 0}$ has infinitesimal generator $\mathcal{G}$ and transition semigroup
$P_t= e^{t\mathcal{G}}$ with symbol $\widehat{P}_t=e^{t \Phi}$. The transition law is written as
$$P_t f_0(x)=\mathbb{E}f_0(Y_t+x),$$
and solves the Cauchy problem
\begin{equation}
\begin{cases}
\frac{\partial f}{\partial t}=\mathcal{G}f,\\
f(x,0)= f_0(x).
\end{cases}
\end{equation}
Then, we have that $P_1 = e^{\mathcal{G}}$.
Let us consider the Fourier transform of \eqref{per},
\begin{equation}
\frac{\partial \widehat{u}}{\partial
t}=-\lambda(I-e^{\Phi(k)})\widehat{u},
\end{equation}
where $\Phi$ is the Fourier multiplier of the operator
$\mathcal{G}$. By integrating with respect to time, we obtain
\begin{equation}\label{peron}
\widehat{u}(k,t)=exp\left(-\lambda t(I-e^{\Phi(k)})\right).
\end{equation}
The characteristic function of the process $(Z_t)_{t\geq 0}$ is given by (see formula \eqref{compa} above)
\begin{align}\label{peron1}
\mathbb{E}e^{ikZ_t}&= exp\left[-\lambda t(I-\mathbb{E}e^{iY(1)k})\right].
\end{align}
Since
$$\mathbb{E}e^{iY(1)k}= e^{\Phi(k)},$$
we have that \eqref{peron1} coincides with \eqref{peron}, as claimed.

\end{proof}

\begin{os}
We specialize formula \eqref{per} in order to obtain some connections with \eqref{per5}.
In the case $\mathcal{G}= -\partial_x$, the Perron-Frobenius operator $K$ is associated to the map
$T: x\mapsto x-1$. Then we have that \eqref{per} becomes
\begin{equation}\label{pois}
\frac{\partial u}{\partial t}= -\lambda(I-e^{-\partial_x})u= \lambda(u(x-1,t)-u(x,t)),
\end{equation}
and $e^{\mathcal{G}}=B$, is the backward operator.
The stochastic solution to \eqref{pois} is therefore
$$Z_t = N(t),$$
that is the homogenous Poisson process.\\
If $\mathcal{G}= -\partial_x^{\alpha}$, that is the Riemann-Liouville derivative of order $\alpha \in (0,1)$ then, by using \eqref{con1} we have that
\begin{equation}
e^{\mathcal{G}}f(x)=e^{-\partial_x^{\alpha}}f(x)= \int_0^\infty ds\, h_\alpha(s, 1) \,
    f(x-s).
\end{equation}
Hence, we have that
$$Y_j\stackrel{d}{=} \mathfrak{H}^\alpha(1),\quad \forall j, $$
so that
\begin{equation}
Z_t\stackrel{d}{=}\sum_{j=1}^{N(t)}\mathfrak{H}_j^\alpha(1).
\end{equation}
Moreover, by using the fact that (see \eqref{koma} above)
$$e^{-t\partial_x^{\alpha}}f(x)=\mathbb{E}e^{-\mathfrak{H}^{\alpha}_t\partial_x}f(x),$$
we have that
\begin{align}
-\lambda(I-e^{-\partial_x^{\alpha}})f(x)&= \lambda \mathbb{E}(e^{-\mathfrak{H}^{\alpha}_t\partial_x}-1)f(x)\\
\nonumber &=\lambda\int_0^{+\infty}\left(e^{-y\partial_x}f(x)-f(x)\right)h_{\alpha}(dy,1)\\
\nonumber &= \lambda\int_0^{+\infty}\left(f(x-y)-f(x)\right)h_{\alpha}(dy,1).
\end{align}
\end{os}

\begin{te}
Let us consider the equation
\begin{equation}\label{per1}
\frac{\partial v}{\partial
t}+\nabla_{\theta}^{\alpha}(\mathbf{u}v)=-\lambda (I-K)v,\quad \mathbf{x}\in\mathbb{R}^d, t>0,
\end{equation}
subject to the initial condition
$v(\mathbf{x},0)=\delta(\mathbf{x})$, where $\alpha \in(0,1)$,
$\mathbf{u}$ is a vector with constant coefficients and $K=
e^{-\mathbf{1}\cdot \nabla}$. The
stochastic solution to \eqref{per1} is given by
\begin{equation}\label{per3}
\mathbf{Y}_t=\mathbf{N}_t+\sum_{j=1}^{d}\boldsymbol{\theta}_j\mathfrak{H}^\alpha\left((\boldsymbol{\theta}_j\cdot
\mathbf{u})t\right),
\end{equation}
where $\mathbf{N}_t=\mathbf{1}N_t$ and $\mathbf{1}=(1,1,\dots,1)$. Furthermore,
\begin{equation}\label{dopo}
v(\mathbf{x},t)= \sum_{m =0}^{\infty}\rho_{\alpha}(\mathbf{x}-m\mathbf{1},t)e^{-\lambda t}\frac{(\lambda t)^m}{m!},
\end{equation}
where $\rho_{\alpha}(\mathbf{x}, t)$ is the fundamental solution of \eqref{main}.
\end{te}

\begin{proof}
The characteristic function of \eqref{per3}, is given by
\begin{align}\label{per4}
\mathbb{E}e^{i\mathbf{k}\cdot\mathbf{Y}}&=
\mathbb{E}exp\left(i\mathbf{k}\cdot\mathbf{1}N_t+\sum_{j=1}^d
i\mathbf{k}\cdot \boldsymbol{\theta}_j
\mathfrak{H}^{\alpha}\left((\boldsymbol{\theta}_j\cdot
\mathbf{u})t\right)\right)\\
\nonumber &= exp\left(-\lambda t
(1-e^{i\mathbf{k}\cdot\mathbf{1}})-\sum_{j=1}^d(\boldsymbol{\theta}_j\cdot
\mathbf{u})(-i\mathbf{k}\cdot
\boldsymbol{\theta}_j)^{\alpha}t\right).
\end{align}
From \eqref{per1}, by taking the Fourier
transform we obtain
\begin{equation}
\frac{\partial \widehat{v}}{\partial t}+\sum_{j=1}^d
(\boldsymbol{\theta}_j\cdot\mathbf{u})(-i\mathbf{k}\cdot\boldsymbol{\theta}_j)^{\alpha}\widehat{v}=-\lambda
(I-e^{i\mathbf{k}\cdot\mathbf{1}})\widehat{v},
\end{equation}
which leads to
\begin{equation}\label{kav}
\widehat{v}(\mathbf{k},t)= exp\left(-\lambda t
(I-e^{i\mathbf{k}\cdot\mathbf{1}})-\sum_{j=1}^d(\boldsymbol{\theta}_j\cdot
\mathbf{u})(-i\mathbf{k}\cdot
\boldsymbol{\theta}_j)^{\alpha}t\right).
\end{equation}
Formula \eqref{kav} coincides with \eqref{per4}, as claimed.\\
In order to prove \eqref{dopo}, we observe that \eqref{kav} can be written as follows
\begin{align}\label{dopo1}
\widehat{v}(\mathbf{k},t)&= exp\left(-\lambda t
(I-e^{i\mathbf{k}\cdot\mathbf{1}})\right)exp\left(-\sum_{j=1}^d(\boldsymbol{\theta}_j\cdot
\mathbf{u})(-i\mathbf{k}\cdot
\boldsymbol{\theta}_j)^{\alpha}t\right)\\
\nonumber &= exp\left(-\lambda t
(I-e^{i\mathbf{k}\cdot\mathbf{1}})\right)\widehat{\rho}_{\alpha}(\mathbf{k},t),
\end{align}
where $\widehat{\rho}_{\alpha}(\mathbf{k},t)$ is the Fourier transform of the fundamental solution of \eqref{main}.
We now consider the Fourier transform of \eqref{dopo}. Recalling the operational rule
$$\rho_{\alpha}(\mathbf{x}-m\mathbf{1},t)= e^{-m(\mathbf{1}\cdot \nabla)}\rho_{\alpha}(\mathbf{x},t),$$
we have
\begin{align}
\widehat{\rho}_{\alpha}(\mathbf{k},t)e^{-\lambda t}\sum_{m=0}^{\infty}e^{i(\mathbf{1}\cdot\mathbf{k})m}\frac{(\lambda t)^m}{m!}=
\widehat{\rho}_{\alpha}(\mathbf{k},t)e^{-\lambda t(1-e^{i\mathbf{k}\cdot\mathbf{1}})},
\end{align}
that coincides with \eqref{dopo1}.
\end{proof}

\begin{os}
We observe that for $\lambda = 0$, we have that
$$v(\mathbf{x},t)= \rho_{\alpha}(\mathbf{x},t),$$
is the fundamental solution of \eqref{main}.
\end{os}

\section{Second order directional derivatives and their fractional power}

We start to deepen the meaning of second order directional
derivative $(\boldsymbol{\theta}\cdot \nabla)^2$. We notice that
\begin{align}\label{ito}
(\boldsymbol{\theta}\cdot \nabla)^2 &= \sum_{i,j}\theta_i \theta_j
\partial_{x_i}\partial_{x_j}\\
\nonumber &=\sum_{i,j}a_{ij}
\partial_{x_i}\partial_{x_j},
\end{align}
where the associated matrix $\{a_{ij}\}$ is  symmetric and singular. Also we assume that $\|\boldsymbol{\theta}\|=1$.\\

The solution to the equation
\begin{equation}\label{calo}
\frac{\partial}{\partial t}u(\mathbf{x},t)
=(\boldsymbol{\theta}\cdot \nabla)^2 u(\mathbf{x},t),\quad
\mathbf{x}\in \mathbb{R}^d, t\geq 0,
\end{equation}
subject to the initial condition
$u(\mathbf{x},0)=\delta(\boldsymbol{\theta}\cdot \mathbf{x})$, is given by (see for example
\cite{Mirko})
\begin{equation}\label{pote2}
u(\mathbf{x},t) = g((\boldsymbol{\theta}\cdot \mathbf{x}),t),
\end{equation}
where
\begin{equation}
g(x,t)=\frac{e^{-\frac{x^2}{4t}}}{\sqrt{4 \pi t}},
\end{equation}
is the law of the one-dimensional Brownian motion $(B_t)_{t>0}$. We
will write $(\mathbf{B}_t)_{t>0}$ for the $d$-dimensional vector,
whose elements are completely correlated one dimensional Brownian motions.
We say that
\begin{equation}
\mathcal{I}_t= \boldsymbol{\theta}\cdot \mathbf{B}_t, \quad t\geq
0,
\end{equation}
is the stochastic solution to \eqref{calo}.
Notice that $\mathcal{I}_t$ is a Gaussian process with singular covariance matrix and degenerate
multivariate normal distribution.

 Therefore, we can write the solution to
\eqref{calo}, subject to an initial condition
$u(\mathbf{x},0)=u_0(\mathbf{x})$, as
\begin{equation}\label{sera}
P_t
u_0(\mathbf{x})=\int_{\mathbb{R}^d}u_0(\mathbf{y})\frac{e^{-\frac{|\boldsymbol{\theta}\cdot(\mathbf{y}-\mathbf{x})|^2}
{4t}}}{\sqrt{4 \pi t}}d\mathbf{y},
\end{equation}
where $P_t=e^{t(\boldsymbol{\theta}\cdot\nabla)^2}$ is the
associated semigroup. We note that: $P_0 = Id$; $P_t 1= 1$ and $P_t P_s f= P_{t+s}f$.

We are now ready to present an integral representation of the power
$\alpha\in (0,1)$ of the operator $(\theta\cdot \nabla)^{2}$ and a stochastic
representation of the related solutions.
\begin{te}
The stochastic solution of the fractional differential equation
\begin{equation}\label{x}
\big(\frac{\partial}{\partial t}
+\left(-(\boldsymbol{\theta}\cdot\nabla)^2\right)^{\alpha}\big)u(\mathbf{x},t)=0,
\quad \mathbf{x} \in \mathbb{R}^d, \, t>0, \alpha \in (0,1),
\end{equation}
subject to the initial condition
$u(\mathbf{x},0)=f(\mathbf{x})\in L^1(\mathbb{R}^d),$
is given by
\begin{equation}\label{st}
\mathcal{I}_t^{\alpha}=\boldsymbol{\theta}\cdot\mathbf{B}_{\mathfrak{H}_t^\alpha}.
\end{equation}
In equation \eqref{x}, the power $\alpha \in (0,1)$ of the operator $(\boldsymbol{\theta}\cdot\nabla)^2$, is given by
\begin{equation}\label{pote}
-\left(-(\boldsymbol{\theta}\cdot\nabla)^2\right)^{\alpha}f(\mathbf{x})=
C(\alpha)\frac{1}{2}\int_{\mathbb{R}^d}\frac{\left(f(\mathbf{y}+\mathbf{x})+f(\mathbf{x}-\mathbf{y})-
2f(\mathbf{x}) \right)}{
|\boldsymbol{\theta}\cdot\mathbf{y}|^{2\alpha+1}}d\mathbf{y},
\end{equation}
with
$$C(\alpha)=
\frac{1}{\pi}\Gamma(2\alpha+1)\sin(\pi\alpha).$$
\end{te}
\begin{proof}
Let us prove \eqref{pote}.
The general
expression for the power $\alpha$ of the operator $\mathcal{A}$ is given by (see
for example \cite{libre, Jacobs})
\begin{equation}\label{pote1}
-(-\mathcal{A})^{\alpha}f(\mathbf{x})=
\frac{\alpha}{\Gamma(1-\alpha)}\int_0^{\infty}\frac{\left(P_sf(\mathbf{x})-f(\mathbf{x})\right)ds}{s^{\alpha+1}},
\end{equation}
where $P_s=e^{s\mathcal{A}}$, is the transition semigroup related to equation \eqref{levo}
with representation \eqref{af} for $\mathcal{A}$.\\
By using equation
\eqref{pote1} and \eqref{sera}, we have that
\begin{align}
-(-(\boldsymbol{\theta}\cdot
\nabla)^2)^{\alpha}f(\mathbf{x})&=\frac{\alpha}{\Gamma(1-\alpha)}\int_0^{\infty}\left(P_s
f(\mathbf{x})-f(\mathbf{x})\right)\frac{ds}{s^{\alpha+1}}\\
\nonumber
&=\frac{\alpha}{\Gamma(1-\alpha)}\int_{\mathbb{R}^d}\left(f(
\mathbf{y})-f(\mathbf{x})
\right)\left[\int_0^{\infty}\frac{e^{-\frac{|\boldsymbol{\theta}\cdot(\mathbf{y}-\mathbf{x})|^2}{4s}}}{\sqrt{4\pi
s}}\frac{ds}{s^{\alpha+1}}\right]d\mathbf{y}\\
\nonumber
&=\frac{4^{\alpha}\alpha}{\Gamma(1-\alpha)}\frac{\Gamma(\alpha+\frac{1}{2})}{\sqrt{\pi}}\int_{\mathbb{R}^d}
\frac{f(\mathbf{y})-f(\mathbf{x})}{
|\boldsymbol{\theta}\cdot(\mathbf{y}-\mathbf{x})|^{2\alpha+1}}d\mathbf{y},
\end{align}

and therefore, we arrive at the following representation
\begin{equation}\label{lapla}
-(-(\boldsymbol{\theta}\cdot
\nabla)^2)^{\alpha}f(\mathbf{x})=C(\alpha)\int_{\mathbb{R}^d}\frac{f(\mathbf{y})-f(\mathbf{x})
}{
|\boldsymbol{\theta}\cdot(\mathbf{y}-\mathbf{x})|^{2\alpha+1}}d\mathbf{y},
\end{equation}
where, in view of the duplication formula
$$\Gamma(2\alpha)=\frac{4^{\alpha-1}}{\sqrt{\pi}}\Gamma(\alpha)\Gamma(\alpha+\frac{1}{2}),$$
we get that
\begin{equation}
C(\alpha)=\frac{4^{\alpha}\alpha}{\Gamma(1-\alpha)}\frac{\Gamma(\alpha+\frac{1}{2})}{\sqrt{\pi}}=
\frac{1}{\pi}\Gamma(2\alpha+1)\sin(\pi\alpha).
\end{equation}
We notice that \eqref{lapla}, must be considered in principal value,
due to the singular kernel. However, we have that
\begin{align}\label{PV}
&p.v \,C(\alpha)
\int_{\mathbb{R}^d}\frac{f(\mathbf{y}+\mathbf{x})-f(\mathbf{x}) }{
|\boldsymbol{\theta}\cdot \mathbf{y}|^{2\alpha+1}}d\mathbf{y}\\
\nonumber
&=C(\alpha)\frac{1}{2}\int_{\mathbb{R}^d}\frac{\left(f(\mathbf{y}+\mathbf{x})+f(\mathbf{x}-\mathbf{y})-
2f(\mathbf{x}) \right)}{
|\boldsymbol{\theta}\cdot\mathbf{y}|^{2\alpha+1}}d\mathbf{y}.
\end{align}
Indeed, taking the Fourier transform of the last term, we obtain
\begin{align}
&C(\alpha)\frac{1}{2}\int_{\mathbb{R}^d}\frac{\left(e^{i\mathbf{k}\cdot\mathbf{y}}+e^{-i\mathbf{k}\cdot\mathbf{y}}-2
\right)}{
|\boldsymbol{\theta}\cdot\mathbf{y}|^{2\alpha+1}}\widehat{f}(\mathbf{k})d\mathbf{y}\\
\nonumber
&=C(\alpha)\frac{1}{2}\int_{\mathbb{R}^d}\frac{\left(e^{i\mathbf{k}\cdot\mathbf{y}}-1\right)+\left(e^{-i\mathbf{k}\cdot\mathbf{y}}-1
\right)}{
|\boldsymbol{\theta}\cdot\mathbf{y}|^{2\alpha+1}}\widehat{f}(\mathbf{k})d\mathbf{y}\\
\nonumber
&=C(\alpha)\int_{\mathbb{R}^d}\frac{\left(e^{i\mathbf{k}\cdot\mathbf{y}}-1\right)}{
|\boldsymbol{\theta}\cdot\mathbf{y}|^{2\alpha+1}}\widehat{f}(\mathbf{k})d\mathbf{y},
\end{align}
which coincides with the Fourier transform of the first term in \eqref{PV}.\\
In order to prove that \eqref{st} is the stochastic solution of \eqref{x},
let us consider
\begin{align}\label{mer}
u(\mathbf{x},t)&=\int_0^{\infty}ds\,h_{\alpha}(s,t)P_s\,f(\mathbf{x})\\
\nonumber &=\int_0^{\infty}ds\, h_{\alpha}(s,t)e^{s\left(\boldsymbol{\theta}\cdot \nabla\right)^2}f(\mathbf{x})\\
\nonumber &=\int_0^{\infty}ds\, h_{\alpha}(s,t)e^{-s\left(-\left(\boldsymbol{\theta}\cdot \nabla\right)^2\right)}f(\mathbf{x})\\
\nonumber &=e^{-t\left(-\left(\boldsymbol{\theta}\cdot\nabla\right)^2\right)^{\alpha}}f(\mathbf{x}).
\end{align}
Then, \eqref{mer} is the solution to \eqref{x}, as claimed.

\end{proof}

\begin{os}
We notice that for $d=1$ and $\alpha \in (0,1)$, the equation
\eqref{pote} becomes
\begin{equation}
-\left(-\left(\frac{\partial}{\partial
x}\right)^2\right)^{\alpha}f(x)=C(\alpha)\int_{\mathbb{R}}\frac{f(y)-f(x)
}{ |y-
x|^{2\alpha+1}}dy=\frac{\partial^{2\alpha}f(x)}{\partial|x|^{2\alpha}},
\end{equation}
that is the Riesz fractional derivative as expected. Also, from equation \eqref{pote}, we find that

\begin{align}\label{lapl}
&\sum_{l=1}^d -(-(\boldsymbol{\theta}_l\cdot
\nabla)^2)^{\alpha}f(\mathbf{x})\\
\nonumber &=
\sum_{l=1}^d\int_{\mathbb{R}^d}\left(f(\mathbf{y})-f(\mathbf{x})\right)J(\boldsymbol{\theta}_l\cdot(\mathbf{x}-\mathbf{y}))d\mathbf{y},
\end{align}
where
\begin{equation}
 J(\boldsymbol{\theta}_l\cdot(\mathbf{x}-\mathbf{y}))=
C(\alpha)\frac{1}{
|\boldsymbol{\theta}_l\cdot(\mathbf{x}-\mathbf{y})|^{2\alpha+1}}.
\end{equation}
If $\boldsymbol{\theta}_i\equiv \mathbf{e}_i$, $i= 1,\cdots,d$ and $\alpha \in (0,1)$, then we have that
\begin{equation}
\sum_{l=1}^d -(-(\mathbf{e}_l\cdot
\nabla)^2)^{\alpha}f(\mathbf{x})= \sum_{l=1}^d \frac{\partial^{2\alpha}}{\partial|x_l|^{2\alpha}}f(\mathbf{x}).
\end{equation}
\end{os}

\begin{os}
Special care must be given to the case $\alpha = \frac{1}{2}$ in
$d=1$. In this case equation \eqref{lapl} becomes a Cauchy integral
\begin{equation}
-\left(-\left(\frac{\partial}{\partial
x}\right)^2\right)^{1/2}f(x)=\frac{p.v}{\pi}\int_{\mathbb{R}}\frac{f(y)-f(x)
}{ |y- x|^{2}}dy,
\end{equation}
where, as usual, "p.v." stands for ``principal value''.
\end{os}

\end{document}